\documentclass{lmcs} 
\pdfoutput=1

\usepackage[utf8]{inputenc}

\usepackage{lastpage}
\lmcsdoi{17}{2}{2}
\lmcsheading{}{\pageref{LastPage}}{}{}%
{Mar.~07,~2020}{Apr.~12,~2021}{}


\usepackage{algorithm} 
\usepackage[noend]{algpseudocode} %

\usepackage{amsthm, stmaryrd, mathtools, amssymb}
\usepackage{comment} 
\usepackage[all]{xy}
\usepackage{tikz}
\usetikzlibrary{graphs, arrows.meta}
\usepackage{scalerel}
\usepackage{mathtools}
\usepackage{color,soul} 
\newcommand\numberthis{\addtocounter{equation}{1}\tag{\theequation}} 

\theoremstyle{definition}\newtheorem{construction}[thm]{Construction}

\newcommand{\powset}{\mathcal{P}}
\newcommand{\var}{\mathit{Var}}
\newcommand{\alphabet}{\mathcal{A}}

\newcommand{\lr}[1]{\left\langle#1\right\rangle}

\renewcommand{\phi}{\varphi}

\newcommand{\ob}{\mathsf{Ob}}

\newcommand{\sets}{\mathbf{Sets}}
\newcommand{\poset}{\mathbf{Poset}}

\newcommand{\clauseC}{\mathcal{C}}

\newcommand{\Prog}{\mathbb{P}} 
\newcommand{\progP}{\mathbb{P}}
\newcommand{\progL}{\mathbb{L}}

\newcommand{\dom}{\mathsf{dom}}

\newcommand{\op}{\mathrm{op}}
\newcommand{\iso}{\cong}

\newcommand{\xto}{\xrightarrow}

\newcommand{\id}{id}
\newcommand{\At}{\mathsf{At}}
\newcommand{\ot}{\leftarrow}
\newcommand{\To}{\Rightarrow}

\newcommand{\limit}{\lim}
\newcommand{\after}{\circ}

\newcommand{\SubDist}{\mathcal{D}_{\scriptscriptstyle \leq 1}}
\renewcommand{\Pr}{\mathsf{Pr}}
\newcommand{\progPr}{\mu}
\newcommand{\Dist}{\mathcal{D}}
\newcommand{\por}{\lor} 

\newcommand{\head}[1]{\mathsf{Head}(#1)}
\newcommand{\body}[1]{\mathsf{Body}(#1)}

\newcommand{\supp}{\mathsf{supp}}

\newcommand{\catC}{\mathbf{C}}

\newcommand{\F}{\mathcal{M}_{{pr}}}
\newcommand{\liftF}{\widehat{\mathcal{M}}_{pr}}

\newcommand{\G}{\SubDist\powset_f}
\newcommand{\liftG}{\lift{\SubDist} \lift{\powset_f}}

\newcommand{\U}{\mathcal{U}}
\newcommand{\K}{\mathcal{K}}

\newcommand{\funcU}{\mathcal{U}}
\newcommand{\funcK}{\mathcal{K}}
\newcommand{\biglor}{\bigvee}
\newcommand{\posw}{\mathsf{pw}} 
\newcommand{\poswtr}{!} 

\renewcommand{\S}{\mathcal{S}}

\newcommand{\intp}[1]{\llbracket#1 \rrbracket}

\newcommand{\Ord}{\mathbf{Ord}} 

\newcommand{\func}[1]{\mathsf{#1}}
\newcommand{\repSubProg}[1]{\mathsf{SubProg}(#1)}
\newcommand{\lab}[1]{\mathsf{Label}(#1)}
\newcommand{\PLP}{\scriptstyle{\mathsf{PLP}}}

\newcommand{\LawTh}{\mathbf{L}}
\newcommand{\lift}[1]{\widehat{#1}}

\newcommand{\genPLP}{\mathcal{S}}

\newcommand{\genPLPclause}{\mathcal{R}}
\newcommand{\genDist}{\mathcal{O}}

\newcommand{\clauseCfunc}{\U \At \times \lift{\powset_f} \U \At}

\newcommand{\clauseCnode}{\bullet}
\newcommand{\probTM}{\mathsf{Pr_\Prog^{\scriptscriptstyle{TM}}}}
\newcommand{\natN}{\mathbb{N}}

\newcommand{\fPowset}{\powset_f}

\newcommand{\treeT}{\mathcal{T}}

\newcommand{\emptyDist}{\varnothing}
\newcommand{\instanceNode}{\blacklozenge}

\newcommand{\pureProg}[1]{|#1|}
\newcommand{\eql}[1]{\stackrel{\mathclap{\mbox{$\scriptstyle #1$}}}{=}}

\newcommand{\progalarm}{{\Prog}^{\scriptscriptstyle al}}
\newcommand{\Atalarm}{{\At}_{\scriptstyle al}}
\newcommand{\palarm}{{p}_{\scriptstyle al}}
\newcommand{\intd}[1]{\langle\!\langle #1 \rangle\!\rangle}

\newcommand{\Sigal}{\Sigma_{\mathsf{al}}}

\newcommand{\LawOp}{\LawTh^\op_\Sigma}
\newcommand{\pclauses}{\widetilde{p}}
\newcommand{\saturate}[1]{{#1}^\sharp}

\newcommand{\problog}{\mathsf{ProbLog}}
\newcommand{\pd}{\mathsf{pD}}
\newcommand{\prism}{\mathsf{PRISM}}
\newcommand{\proves}{\Rightarrow}
\newcommand{\Acc}{\mathsf{Acc}}
\newcommand{\Rej}{\mathsf{Rej}}

\newcommand{\WLP}{\scriptstyle{\mathsf{WLP}}}
\newcommand{\semiK}{\mathbb{K}}
\newcommand{\semiPlus}{\oplus}
\newcommand{\semiTimes}{\otimes}
\newcommand{\semiSum}{\bigoplus}
\newcommand{\semiProd}{\bigotimes}
\newcommand{\progW}{\mathbb{W}}
\newcommand{\plusUnit}{\mathbf{0}}
\newcommand{\prodUnit}{\mathbf{1}}
\newcommand{\progShortPath}{\mathbb{P}^{sp}}
\newcommand{\progGroundSP}{\mathbb{P}^{gsp}}
\newcommand{\nonNegReal}{\mathbb{R}_{\geq 0}}
\newcommand{\weight}{\omega}
\newcommand{\multiFunc}{\mathcal{M}}
\newcommand{\wFinalSem}[2]{\llbracket #1 \rrbracket_{#2}}

\newcommand{\wFinalSemGen}[2]{\llbracket #1 \rrbracket_{#2}}
\newcommand{\cMultiFunc}{\multiFunc_{c}}

\newcommand{\coalgU}{\mathsf{u}}

  \newcommand{\whitespacearound}[1]{%
    \begin{tikzpicture}[baseline=(mainnode.base)]
      \node[inner sep=2pt,outer sep=0pt] (mainnode) {%
        \ensuremath{\scriptstyle #1}};
    \end{tikzpicture}%
  }

\begin{document}

\title[Coalgebraic Semantics for Probabilistic Logic Programming]{Coalgebraic Semantics \\ for Probabilistic Logic Programming} 

\author[T.~Gu]{Tao GU}	
\address{University College London}	
\email{tao.gu.18@ucl.ac.uk}  
\email{f.zanasi@ucl.ac.uk}  

\author[F. Zanasi]{Fabio Zanasi}	





\begin{abstract}
  \noindent Probabilistic logic programming is increasingly important in artificial intelligence and related fields as a formalism to reason about uncertainty. It generalises logic programming with the possibility of annotating clauses with probabilities. This paper proposes a coalgebraic semantics on probabilistic logic programming. Programs are modelled as coalgebras for a certain functor $\func{F}$, and two semantics are given in terms of cofree coalgebras. First, the cofree $\func{F}$-coalgebra yields a semantics in terms of derivation trees. Second, by embedding $\func{F}$ into another type $\func{G}$, as cofree $\func{G}$-coalgebra we obtain a `possible worlds' interpretation of programs, from which one may recover the usual distribution semantics of probabilistic logic programming. Furthermore, we show that a similar approach can be used to provide a coalgebraic semantics to weighted logic programming.
\end{abstract}

\maketitle

\section{Introduction}\label{sec:intro}

Probabilistic logic programming ($\PLP$) \cite{ProbLogProgPioneerI,ProbLogProgPioneerII,sato1995statistical} is a family of approaches extending the declarative paradigm of logic programming with the possibility of reasoning about uncertainty. This has been proven useful in various applications, including bioinformatics \cite{DeRaedt:2007:BioinformaticsI,Holmes:BioinformaticsII}, robotics \cite{thrun2005probabilistic} and the semantic web \cite{Zese:2017:SemanticWeb}.

The most common version of $\PLP$ --- on which for instance $\problog$ is based \cite{DeRaedt:2007:BioinformaticsI}, the probabilistic analogue of $\mathsf{Prolog}$ --- is defined by annotating clauses in programs with mutually independent probabilities. As for the interpretation, \emph{distribution semantics} \cite{sato1995statistical} is typically used as a benchmark for the various implementations of $\PLP$, such as $\pd$, $\prism$ and $\problog$ \cite{riguzzi2014}. While in logic programming the central task is whether a goal is provable using the clauses in the program as axioms, for $\PLP$ the typical question one asks is what is the probability of a goal being provable. In distribution semantics, all the clauses are seen as independent random events, and such probability is obtained as the sum of the probabilities of all the \emph{possible worlds} (sets of clauses) in which the goal is provable. The distribution semantics is particularly interesting because it is compatible with the encoding of Bayesian networks as probabilistic logic programs \cite{riguzzi2014}, thus indicating that $\PLP$ can be effectively employed for Bayesian reasoning.

The main goal of this work is to present a coalgebraic perspective on $\PLP$ and its distribution semantics. We first consider the case of ground programs (Section \ref{sec:groundcase}), that is, programs without variables. Our approach is based on the observation --- inspired by the coalgebraic treatment of `pure' logic programming \cite{komendantskaya2010coalgebraic} --- that ground programs are in 1-1 correspondence with coalgebras for the functor $\F\powset_f$, where $\F$ is the finite multiset functor on $[0,1]$ and $\powset_f$ is the finite powerset functor. We then provide two coalgebraic semantics for ground $\PLP$.
\begin{itemize}
\item 	The first interpretation $\intp{-}$ is in terms of execution trees called \emph{stochastic derivation trees}, which represent parallel SLD-derivations of a program on a goal. Stochastic derivation trees are the elements of the cofree $\F\powset_f$-coalgebra on a given set of atoms $\At$, meaning that any goal $A \in \At$ can be given a semantics in terms of the corresponding stochastic derivation tree by the universal map $\intp{-}$ to the cofree coalgebra.
\item The second interpretation $\intd{-}$ recovers the usual distribution semantics of $\PLP$. This requires some work, as expressing probability distributions on the possible worlds needs a different coalgebra type. We introduce \emph{distribution trees}, a tree-like representation of the distribution semantics, as the elements of the cofree $\G\powset_f$-coalgebra on $\At$, where $\SubDist$ is the sub-probability distribution monad. In order to characterise $\intd{-}$ as the map given by universal property of distribution trees, we need a canonical extension of $\PLP$ to the setting of $\G\powset_f$-coalgebras. This is achieved via a `possible worlds' natural transformation $\F\powset_f \To \G\powset_f$.
\end{itemize}

In the second part of the paper (Section \ref{sec:general_case}) we recover the same framework for arbitrary probabilistic logic programs, possibly including variables. The encoding of programs as coalgebras is subtler. The space of atoms is now a presheaf indexed by a `Lawvere theory' of terms and substitutions. The coalgebra map can be defined in different ways, depending on the substitution mechanism on which one wants to base resolution. For pure logic programs, the definition by term-matching is the best studied, with~\cite{komendantskaya2011coalgebraic} observing that moving from sets to posets is required in order for the corresponding coalgebra map to be well-defined as a natural transformation between presheaves. A different route is taken in~\cite{BonchiZ15}, where the problem of naturality is neutralised via `saturation', a categorical construction which amounts to defining resolution by unification instead of term-matching.

In developing a coalgebraic treatment of $\PLP$ with variables, we follow the saturation route, as it also allows to recover the term-matching approach, via `desaturation'~\cite{BonchiZ15}. This provides a cofree coalgebra semantics $\intp{-}$ for arbitrary $\PLP$ programs, as a rather straightforward generalisation of the saturated semantics for pure logic programs. On the other hand, extending the ground distribution semantics $\intd{-}$ to arbitrary $\PLP$ programs poses some challenges: we need to ensure that, in computing the distribution over possible worlds associated to each sub-goal in the computation, each clause of the program is `counted' only once. This is solved by tweaking the coalgebra type of the distribution trees for arbitrary $\PLP$ programs, so that some nodes are labelled with clauses of the program. Thanks to this additional information, the term-matching  distribution semantics of an arbitrary $\PLP$ goal is computable from its distribution tree.

After developing our framework for $\PLP$, in the last part of the paper (Section \ref{sec:WLP}) we show how the same approach yields a coalgebraic semantics for weighted logic programming ($\WLP$). $\WLP$ generalises standard logic programming by adding weights to clauses; it is mostly used in the specification of dynamic programming algorithms in various fields, including natural language processing \cite{eisner2007} and computational biology \cite{durbin1998}. Our coalgebraic description yields a derivation semantics $\intp{-}$ both for ground and arbitrary $\WLP$ programs, from which one may compute the weight associated with a goal.


In light of the coalgebraic treatment of pure logic programming~\cite{komendantskaya2010coalgebraic,komendantskaya2011coalgebraic,bonchi2013saturated,BonchiZ15}, the generalisation to $\PLP$ and $\WLP$ may not appear so surprising. In fact, we believe its importance is two-fold. First, whereas the derivation semantics $\intp{-}$ is a straight generalisation of the pure setting, the distribution semantics $\intd{-}$ is genuinely novel, and does not have counterparts in pure logic programming. Second, this work provides a starting point for a generic coalgebraic treatment of variations of logic programming:
\begin{itemize}
	\item The coalgebraic approach to pure logic programming has been used as a formal justification \cite{KomendantskayaP16,KomendantskayaL17} for coinductive logic programming \cite{Komendantskaya18PositionPaper,GuptaCoinductiveLogicProgramming}. Coinduction in the context of probabilistic and weighted logic programs is, to the best of our knowledge, a completely unexplored field, for which the current paper establishes semantic foundations.
	\item As mentioned, reasoning in Bayesian networks can be seen as a particular case of $\PLP$, equipped with the distribution semantics. Our coalgebraic perspective thus readily applies to Bayesian reasoning, paving the way for combination with recent works \cite{JacobsKZ19,JacobsZ19, DahlqvistDGK16} modelling belief revision, causal inference and other Bayesian tasks in algebraic terms.
\end{itemize}

\noindent We leave the exploration of these venues as follow-up work.

\medskip

\noindent This work extends the conference paper \cite{GuZ19} with the addition of the missing proofs (Appendix \ref{app:missing-proofs}), and novel material on weighted logic programming (Section \ref{sec:WLP} and Appendix~\ref{app:compute-weight}).

\section{Preliminaries}\label{sec:preliminaries}

\noindent\textbf{Signature, Terms, and Categories.} A \textit{signature} $\Sigma$ is a set of function symbols, each equipped with a fixed finite arity. Throughout this paper we fix a signature $\Sigma$, and a countably infinite set of variables $\var = \{x_1, x_2,\dots\}$. The $\Sigma$-terms over $\var$ are defined as usual. A \textit{context} is a finite sequence of variables $\lr{x_1, x_2,\dots, x_n }$. With some abuse of notation, we shall often use $n$ to denote this context. We say a $\Sigma$-term $t$ is \emph{compatible} with context $n$ if the variables appearing in $t$ are all contained in $\{x_1,\dots, x_n \}$.

We are going to reason about $\Sigma$-terms categorically using Lawvere theories. First, we will use $\ob(\catC)$ to denote the set of objects and $\catC[C,D]$ for the set of morphisms $C\to D$ in a category $\catC$. A $\catC$-indexed \textit{presheaf} is a functor $\func{F} \colon \catC\to \sets$. $\catC$-indexed presheaves and natural transformations between them form a category $\sets^\catC$. Recall that the (opposite) Lawvere Theory of $\Sigma$ is the category $\LawTh^\op_\Sigma$ with objects the natural numbers and morphisms $\LawTh^\op_\Sigma[n,m]$ the n-tuples $\lr{t_1,\dots, t_n}$, where each $t_i$ is a $\Sigma$-term in context $m$.  For modelling logic programming, it is convenient to think of each $n\in \ob(\LawTh^\op_\Sigma)$ as representing the context $\lr{x_1,\dots, x_n }$, and a morphism $\lr{t_1,\dots, t_n} \colon n\to m$ as the substitution transforming $\Sigma$-terms in context $n$ to $\Sigma$-terms in context $m$ by replacing each $x_i$ with $t_i$. For this reason we shall also refer to $\LawTh^\op_\Sigma$ morphisms simply as substitutions (notation $\theta,\tau,\sigma, \dots$).

\smallskip
\noindent\textbf{Logic programming.} We now recall the basics of (pure) logic programming, and refer the reader to \cite{Lloyd87} for a more systematic exposition. An \emph{alphabet} $\alphabet$ consists of a signature $\Sigma$, a set of variables $\var$, and a set of predicate symbols $\{P_1, P_2,\dots\}$, each with a fixed finite arity. Given an $n$-ary predicate symbol $P$ in $\alphabet$, and $\Sigma$-terms $t_1,\dots, t_n$, $P(t_1\cdots t_n)$ is called an \emph{atom} over $\alphabet$. We use $A,B,\dots$ to denote atoms. Given an atom $A$ in context $n$, and a substitution $\theta = \lr{t_1,\dots, t_n} \colon n\to m$, we write $A\theta$ for the \textit{substitution instance} of $A$ obtained by replacing each appearance of $x_i$ with $t_i$ in $A$. For convenience, we also use $\{B_1,\dots, B_k\}\theta$ as a shorthand for $\{B_1\theta, \dots, B_k\theta\}$. Given two atoms $A$ and $B$ (over $\alphabet$), a \textit{unifier} of $A$ and $B$ is a pair $\lr{\sigma, \tau}$ of substitutions such that $A\sigma = B\tau$. \textit{Term matching} is a special case of unification, where $\sigma$ is the identity substitution. In this case we say that $\tau$ matches $B$ with $A$ if $A = B\tau$.

A (pure) logic program $\progL$ consists of a finite set of clauses $\clauseC$ in the form $H\ot B_1,\dots, B_k$, where $H,B_1,\dots, B_k$ are atoms. $H$ is called the \emph{head} of $\clauseC$, and $B_1,\dots, B_k$ form the \emph{body} of $\clauseC$. We denote $H$ by $\head{\clauseC}$, and $\{B_1,\dots, B_k\}$ by $\body{\clauseC}$. A clause $\clauseC$ with empty body is also called a \emph{fact}. A \emph{goal} is simply an atom, and we use this terminology in the context of certain logic programming reasoning tasks. Since one can regard a clause $H\ot B_1,\dots, B_k$ as the logic formula $B_1 \land \cdots \land B_k \to H$, we say that a goal $G$ is \emph{derivable} in $\progL$ if there exists a derivation of $G$ with empty assumption using the clauses in $\progL$.

The central task of logic programming is to check whether an atom $G$ is \emph{provable} in a program $\progL$, in the sense that some substitution instance of $G$ is derivable in $\progL$. The key algorithm for this task is SLD-resolution, see e.g. \cite{Lloyd87}. We use the notation $\progL\vdash G$ to mean that $G$ is provable in $\progL$.

\smallskip
\noindent\textbf{Probabilistic logic programming.} We now recall the basics of $\PLP$; the reader may consult \cite{de2007problog, DeRaedt:2007:BioinformaticsI} for a more comprehensive introduction. A probabilistic logic program $\Prog$ based on a logic program $\progL$ assigns a probability label $r$ to each clause $\clauseC$ in $\progL$, denoted as $\lab{\clauseC}$. One may also regard $\Prog$ as a set of probabilistic clauses of the form $r::\clauseC$, where $\clauseC$ is a clause in $\progL$, and each clause $\clauseC$ is assigned a unique probability label $r$ in $\Prog$. We also refer to $r::\clauseC$ simply as clauses. 

\begin{exa}\label{ex:ground}
As our leading example we introduce the following probabilistic logic program $\progalarm$. It models the scenario of Mary's house alarm, which is supposed to detect burglars, but it may be accidentally triggered by an earthquake. Mary may hear the alarm if she is awake, but even if the alarm is not sounding, in case she experiences an auditory hallucination (paracusia). The language of $\progalarm$ includes $0$-ary predicates $\sf{Alarm}$, $\sf{Eearthquake}$, $\sf{Burglary}$, and $1$-ary predicates $\sf{Wake}(-)$, $\sf{Hear\_alarm}(-)$ and $\sf{Paracusia}(-)$, and signature $\Sigal = \{\mathsf{Mary}^0\}$ consisting of a constant. We do not have variables here, so $\progalarm$ is a ground program. For readability we abbreviate $\mathsf{Mary}$ as $\mathsf{M}$ in the program.

\medskip

\begin{tabular}{rll|rll}
$0.01 :: $ & $\sf{Earthquake}$ & $\ot $ &
$0.01 :: $ & $\sf{Paracusia(M)}$ & $\ot $  \\
$0.2 ::$ & $\sf{Burglary}$ & $\ot$&
$0.6 :: $ & $\sf{Wake(M)}$ & $\ot$ \\
$0.5:: $ & $\sf{Alarm}$ & $\ot \sf{Earthquake}$ &
$0.8 :: $ & $\sf{Hear\_alarm(M)}$ & $\ot \sf{Alarm},~ \sf{Wake(M)}$ \\
$0.9 :: $ & $\sf{Alarm}$ & $\ot \sf{Burglary}$ &
$0.3 :: $ & $\sf{Hear\_alarm(M)}$ & $\ot \sf{Paracusia(M)}$
\end{tabular}
\end{exa}

As a generalisation of the pure case, in probabilistic logic programming one is interested in the \emph{probability} of a goal $G$ being provable in a program $\Prog$. There are potentially multiple ways to define such probability --- in this paper we focus on Sato's \textit{distribution semantics} \cite{sato1995statistical} as below.

Given a probabilistic logic program $\Prog = \{p_1:: \clauseC_1, \dots, p_n:: \clauseC_n \}$, let $\pureProg{\Prog}$ be its underlying pure logic program, namely $\pureProg{\Prog} = \{\clauseC_1,\dots, \clauseC_n\}$. A \emph{sub-program} $\progL$ of $\pureProg{\Prog}$ is a logic program consisting of a subset of the clauses in $\pureProg{\Prog}$. This justifies using $\powset(\pureProg{\Prog})$ to denote the set of all sub-programs of $\pureProg{\Prog}$, and using $\progL\subseteq \pureProg{\Prog}$ to denote that $\progL$ is a sub-program of $\Prog$. The central concept of the distribution semantics is that $\Prog$ determines a distribution $\mu_{\Prog}$ over the sub-programs $\powset(\pureProg{\Prog})$: for any $\progL\in \powset(\pureProg{\Prog})$, 
\[
\progPr_{\Prog}(\progL) \coloneqq \left( \prod_{\clauseC_i \in \progL} p_i \right) \cdot \left( \prod_{\clauseC_j \in \pureProg{\Prog} \setminus \progL} (1-p_j) \right)
\]
where each $p_i$ is the probability label of the clause $\clauseC_i$ in $\progP$. We simply refer to the value $\mu_{\Prog}(\progL)$ as the \textit{probability} of the sub-program $\progL$. For an arbitrary goal $G\in \At$, the \textit{success probability} (or simply the probability) $\Pr_\Prog(G)$ of $G$ w.r.t. program $\Prog$ is then defined as the sum of the probabilities of all the sub-programs of $\Prog$ in which $G$ is provable:

\begin{equation}\label{eq:distrsemantics}
\Pr_\Prog(G)  := \sum_{\pureProg{\Prog} \supseteq \progL \, \vdash \, G} \progPr_\Prog(\progL)  = \sum_{\pureProg{\Prog} \supseteq \progL \, \vdash \, G} \left( \prod_{\clauseC_i \in \progL} p_i \cdot \prod_{\clauseC_j \in \pureProg{\Prog} \setminus \progL} (1-p_j) \right)
\end{equation}
Intuitively one can regard every clause in $\Prog$ as a random event, then every sub-program $\progL$ can be seen as a possible world over these events, and $\mu_\Prog$ is a distribution over all the possible worlds. The success probability $\Pr_\Prog(G)$ is then simply the probability of all the possible worlds in which the goal $G$ is provable.
\begin{exa}
For the program $\progalarm$, consider the goal $\sf{Hear\_alarm(M)}$. By the definition of distribution semantics \eqref{eq:distrsemantics}, we can compute that the success probability $\Pr_{\progalarm}(\sf{Hear\_alarm(M)})$ is $\approx 0.0911$.
\end{exa}

\section{Ground case}\label{sec:groundcase}
In this section we introduce a coalgebraic semantics for \emph{ground} probabilistic logic programming, i.e. for those programs where no variable appears. Our approach consists of two parts. First, we represent $\PLP$ programs as coalgebras (Subsection \ref{sec:coalg_groundcase}) and their executions as a final coalgebra semantics (Subsection \ref{sec:finalcoalg_groundcase}) --- this is a straightforward generalisation of the coalgebraic treatment of pure logic programs given in \cite{komendantskaya2010coalgebraic}. Next, in Subsection \ref{sec:distrsem_groundcase} we investigate how to represent the distribution semantics as a final coalgebra, via a transformation of the coalgebra type of $\PLP$ programs. Appendix \ref{sec:algo_ground} shows how the success probability of a goal is effectively computable from the above representation.

\subsection{Coalgebraic Representation of $\PLP$}\label{sec:coalg_groundcase}
A ground $\PLP$ program can be represented as a coalgebra for the composite $\F \powset_f\colon \sets \to \sets$ of the finite probability functor $\F  \colon \sets \to \sets$ and the finite powerset functor $\powset_f  \colon \sets \to \sets$. The definition of $\F$ deserves some further explanation. It can be seen as the finite multiset functor based on the commutative monoid $([0,1], 0, \por)$, where $\por$ is the probabilistic `or' for independent events defined as $a\por b \coloneqq 1- (1-a)(1-b)$. That is to say, on objects, $\F(A)$ is the set of all \textit{finite probability assignments} $\phi \colon A\to [0,1]$, namely those $\phi$ with a finite support $\supp(\phi) \coloneqq \{a\in A \mid \phi(a) \neq 0\}$. For $\phi$ with  support $\{a_1,\dots, a_k\}$ and values $\phi(a_i) = r_i$, it will often be convenient to use the standard notation $\phi = r_1 a_1 + \cdots + r_k a_k$ or $\phi = \sum_{i=1}^k r_i a_i$, where the purely formal `$+$' and `$\sum$' here should not be confused with the arithmetic addition. On morphisms, $\F(h \colon A\to B)$ maps $\sum_{i=1}^k r_i a_i$ to $\sum_{i=1}^k r_i h(a_i)$.

Fix a ground probabilistic logic program $\Prog$ on a set of ground atoms $\At$. The definition of $\Prog$ can be encoded as an $\F \powset_f$-coalgebra $p \colon \At\to \F \powset_f(\At)$, as follows. Given $A\in \At$,
\begin{equation*}
\begin{array}{ccl}
 	p(A)  \colon  & \powset_f(\At) & \to \ [0,1] \\[0.5em]
 	& \{B_1,\dots, B_n\} & \mapsto  \
 	\begin{cases}
 		r \quad \text{if } r:: A \ot B_1,\dots, B_n \text{ is a clause in } \Prog \\
 	   0 \quad \text{otherwise.}
	\end{cases}
\end{array}
\end{equation*}
Or, equivalently,
$p(A) \coloneqq \sum\limits_{(r:: A \ot B_1,\dots, B_n) \ \in \ \Prog} r \{B_1,\dots, B_n\}$. Note that each $p(A)$ has a finite support because the program $\progP$ consists of finitely many clauses.

\begin{exa}
Consider program $\progalarm$ from Example~\ref{ex:ground}. The set of ground atoms $\Atalarm$ is $\{\sf{Alarm}, \sf{Earthquake},$ $ \sf{Burgary},\sf{Wake(M)}, \sf{Paracusia(M)}, \sf{Hear\_alarm(M)}\}$. Here are some values of the corresponding coalgebra $\palarm \colon \Atalarm \to \F \powset_f \Atalarm$:
\begin{align*}
\hspace{-.5cm}\palarm(\mathsf{Hear\_alarm(M)}) & = 0.8 \{\mathsf{Alarm}, \mathsf{Wake(M)}\} + 0.3 \{\mathsf{Paracusia(M)}\} \\
\palarm(\mathsf{Earthquake}) & = 0.01 \{\}
\end{align*}
\end{exa}

\begin{rem}
One might wonder why not simply adopt $\powset_f (\powset_f(-) \times [0,1])$ as the coalgebra type for $\PLP$. The reason is that, although every ground $\PLP$ program generates a $\powset_f (\powset_f(-) \times [0,1])$-coalgebra, such encoding fails to be a 1-1 correspondence: a clause $\clauseC \in \powset_f(\At)$ may be associated with different values in $[0,1]$, which violates the standard definition of $\PLP$ programs.
\end{rem}

\subsection{Derivation Semantics}\label{sec:finalcoalg_groundcase}
In this section we are going to construct the final $\At\times \F \powset_f (-)$-coalgebra, thus providing a semantic interpretation for probabilistic logic programs based on $\At$.

Before the technical developments, we give an intuitive view on the semantics that the final coalgebra is going to provide. We shall represent each goal as a \emph{stochastic derivation tree} in the final $\At\times \F \powset_f (-)$-coalgebra. These trees are the probabilistic version of and-or derivation trees, which represent parallel SLD-resolutions for pure logic programming \cite{gutpa1994}. One can view a stochastic derivation trees as the unfolding of a goal under a $\PLP$ program. 

\begin{defi}[Stochastic derivation trees] \label{def:stochastictrees_ground}
Given a ground $\PLP$ program $\Prog$ based on $\At$, and an atom $A\in \At$, the \emph{stochastic derivation tree} for $A$ in $\Prog$ is the possibly infinite tree $\mathcal{T}$ such that:
\begin{enumerate}
	\item Every node is either an atom-node (labelled with an atom $A' \in \At$) or a clause-node (labelled with $\bullet$). These two types of nodes appear alternatingly in depth, in this order. In particular, the root is an atom-node labelled with $A$.
	\item Each edge from an atom-node to its (clause-)children is labelled with a probability value.
    \item Suppose $s$ is an atom-node with label $A'$. Then for every clause $r:: A' \ot B_1,\dots, B_k$ in $\Prog$, $s$ has exactly one child $t$ such that the edge $s \to t$ is labelled with $r$, and $t$ has exactly $k$ children labelled with $B_1,\dots, B_k$, respectively.
\end{enumerate}
\end{defi}

The final coalgebra semantics $\intp{-}_p$ for a program $\Prog$ will map a goal $A$ to the stochastic derivation tree representing all possible SLD-resolutions of $A$ in $\Prog$.

\begin{exa}\label{example_base_and-or}
Continuing Example \ref{ex:ground}, $\intp{{\sf{Hear\_alarm(M)}}}_{\palarm}$ is the stochastic derivation tree below. The subtree highlighted in red represents one of the successful proofs of $\sf{Hear\_alarm(M)}$ in $\palarm$: indeed, note that a single child is selected for each atom-node $A$ (corresponding to a clause matching $A$), all children of any clause-node are selected (corresponding to the atoms in the body of the clause), and the subtree has clause-nodes as leaves (all atoms are proven).
\begin{equation}
	\label{eq:exampletreeground}
\vcenter{
\xymatrix@R=5pt{
&&& \sf{Hear\_Alarm(M)} \ar@{-}@[red][dl]_{0.8} \ar@{-}[dr]^{0.3}  & \\
&& \bullet \ar@{-}@[red][dl] \ar@{-}@[red][dr] && \bullet \ar@{-}[d] \\
& \sf{Alarm} \ar@{-}@[red][dl]_{0.5} \ar@{-}[dr]^{0.9} && \sf{Wake(M)} \ar@{-}@[red][d]_{0.6} & \sf{Paracusia(M)} \ar@{-}[d]^{0.01} \\
\bullet \ar@{-}@[red][d] && \bullet\ar@{-}[d] & \bullet & \bullet \\
\sf{Earthquake} \ar@{-}@[red][d]_{0.01} && \sf{Burglary} \ar@{-}[d]_{0.2} && \\
\bullet && \bullet &&
}
}
\end{equation}
Any such subtree describes a proof, but does not yield a probability value to be associated to a goal --- this is the remit of the distribution semantics, see Example \ref{ex:distrsemantics_ground} below.
\end{exa}

In the remaining part of the section, we construct the cofree coalgebra for $\F \powset_f$ via a so-called terminal sequence \cite{worrell1999terminal}, and obtain $\intp{-}_p$ from the resulting universal property. We report the steps of the terminal sequence as they are instrumental in showing that the elements of the cofree coalgebra can be seen as stochastic derivation trees.



\begin{construction}\label{constr:terminalseq_ground}
The terminal sequence for the functor $\At \times \F \powset_f (-): \sets\to \sets$ consists of sequences of objects $\{X_\alpha\}_{\alpha\in \Ord}$ and arrows $\{\delta^\alpha_\beta \colon X_\alpha \to X_\beta\}_{\beta<\alpha\in \Ord}$ constructed by the following transfinite induction:
\begin{align*}
X_\alpha & :=
\begin{cases}
\At & \alpha = 0 \\
\At \times \F \powset_f (X_\xi) & \alpha = \xi+1 \\
\limit\{\delta^\chi_\xi \mid \xi < \chi < \alpha\} & \alpha ~\text{is limit}~
\end{cases} \\
\delta^\alpha_\beta & :=
\begin{cases}
\pi_1 & \alpha = 1, \beta = 0 \\
\id_\At \times \F \powset_f (\delta^{\xi+1}_{\xi}) & \alpha = \beta + 1 = \xi + 2 \\
\text{the limit projections} & \alpha~\text{is limit}, \beta<\alpha\\
\text{universal map to $X_{\beta}$} & \beta~\text{is limit}, \alpha = \beta + 1
\end{cases}
\end{align*}
\end{construction}

\begin{prop}\label{cor:terminalseq_ground_ok}
The terminal sequence for the functor $\At \times \F \powset_f(-)$ converges to a limit $X_\gamma$ such that $X_\gamma \iso X_{\gamma+1}$.
\end{prop}
\begin{proof} We apply the following result from \cite{worrell1999terminal}:
\begin{propC}[{\cite[Corollary 3.3]{worrell1999terminal}}]\label{prop:worrell-cor}
If $T$ is an accessible endofunctor on a locally presentable category, and if $T$ preserves monics, then the terminal $T$-sequence $\{A_\alpha, f^\alpha_\beta\}$ converges, necessarily to a terminal $T$-coalgebra.
\end{propC}
Since $\sets$ is a locally presentable category, it remains to verify the two conditions for $\F \fPowset$. It is well-known that $\powset_f$ is $\omega$-accessible, and $\F$ has the same property, see e.g. \cite[Prop. 6.1.2]{SilvaThesis}. Because accessibility is defined in terms of colimit preservation, it is clearly preserved by composition, and thus $\F\powset_f$ is also accessible. It remains to check that it preserves monics. For $\F$, given any monomorphism $i: C\to D$ in $\sets$, suppose $\F(i)(\phi) = \F(i)(\phi')$ for some $\phi,\phi'\in \F(C)$. Then for any $d\in D$, $\F(i)(\phi)(d) = \F(i)(\phi')(d)$. If we focus on the image $i[C]$, then there is an inverse function $i^{-1}: i[C]\to C$, and $\F(i)(\phi)= \F(i)(\phi')$ implies that $\phi(i^{-1}(d)) = \phi'(i^{-1}(d))$ for any $d\in i[C]$. But this simply means that $\phi = \phi'$. As the same is true for $\powset_f$ and the property is preserved by composition, we have that $\F \powset_f$ preserves monics. Therefore we can conclude that the terminal sequence for $At \times \F\powset_f$ converges to the cofree $\F\powset_f$-coalgebra on $\At$.
\end{proof}


Note that $X_{\gamma +1}$ is defined as $\At \times \F \powset_f (X_{\gamma})$, and the above isomorphism makes $X_{\gamma} \to \At \times \F \powset_f (X_{\gamma})$ the final $\At \times \F \powset_f$-coalgebra --- or, in other words, \emph{cofree $\F \powset_f$-coalgebra} on $\At$. As for the tree representation of the elements of $X_{\gamma}$, recall that elements of the cofree $\powset_f \powset_f$-coalgebra on $\At$ can be seen as and-or trees \cite{komendantskaya2010coalgebraic}. By replacing the first $\powset_f$ with $\F$, effectively one adds probability values to the edges from and-nodes to or-nodes (which are edges from atom-nodes to clause-nodes in our stochastic derivation trees), as in \eqref{eq:exampletreeground}. Thus stochastic derivation trees as in Definition~\ref{def:stochastictrees_ground} are elements of $X_{\gamma}$. The action of the coalgebra map $\iso \colon X_{\gamma} \to \At \times \F \powset_f(X_{\gamma})$ is best seen with an example: the tree $\mathcal{T}$ in \eqref{eq:exampletreeground} (as an element of $X_{\gamma}$) is mapped to the pair $\lr{\mathsf{Hear\_alarm(M)},\phi}$, where $\phi$ is the function $\powset_f(X_{\gamma}) \to [0,1]$ assigning $0.8$ to the set consisting of the subtrees of $\mathcal{T}$ with root $\mathsf{Alarm}$ and with root $\mathsf{Wake(M)}$, $0.3$ to the singleton consisting to the subtree of $\mathcal{T}$ with root $\mathsf{Paracusia(M)}$, and $0$ to any other finite set of trees.

With all the definitions at hand, it is straightforward to check that $\intp{-}_p$ mapping $A \in \At$ to its stochastic derivation tree in $\progP$ makes the following diagram commute
\begin{equation*}
\xymatrix@R=15pt@C+75pt{
\At \ar@{-->}[r]^{\intp{-}_p} \ar[d]^{<id,p>} & X_{\gamma} \ar[d]^{\iso} \\
\At \times \F \powset_f(\At) \ar[r]^{id \times \F \powset_f(\intp{-}_p)} & \At \times \F \powset_f(X_{\gamma})
}
\end{equation*}
 and thus by uniqueness it coincides with the $\At \times \F \powset_f$-coalgebra map provided by the universal property of the final $\At \times \F \powset_f$-coalgebra $X_{\gamma} \to \At \times \F \powset_f (X_\gamma)$.

\subsection{Distribution Semantics}\label{sec:distrsem_groundcase}
This section gives a coalgebraic representation of the usual \emph{distribution semantics} of probabilistic logic programming. As in the previous section, before the technical developments we gather some preliminary intuition. Recall from Section \ref{sec:preliminaries} that the core of the distribution semantics is the probability distribution over the sub-programs (subsets of clauses) of a given program $\Prog$. These sub-programs are also called (possible) worlds, and the distribution semantics of a goal is the sum of the probabilities of all the worlds in which it is provable.

In order to encode this information as elements of a final coalgebra, we need to present it in tree-shape. Roughly speaking, we form a distribution over the sub-programs along the execution tree. This justifies the following notion of \emph{distribution trees}.

\begin{defi}[Distribution trees] \label{def:distrtree_ground}
Given a $\PLP$ program $\Prog$ and an atom $A$, the \emph{distribution tree} for $A$ in $\Prog$ is the possibly infinite tree $\mathcal{T}$ satisfying the following properties:
\begin{enumerate}
	\item Every node is exactly one of the three kinds: atom-node (labelled with an atom $A' \in \At$), world-node (labelled with $\circ$), clause-node (labelled with $\bullet$). They appear alternatingly in this order in depth. In particular the root is an atom-node labelled with $A$.
	\item Every edge from an atom-node to its (world-)children is labelled with a probability value, and the labels on all the edges (if exist) from the same atom-node sum up to $1$.
	\item Suppose $s$ is an atom-node labelled with $A'$, and $C = \{ \clauseC_1,\dots, \clauseC_m \}$ is the set of all the clauses in $\Prog$ whose head is $A'$. Then $s$ has $2^m$ (world-)children, each standing for a subset $X$ of $C$. If a child $t$ stands for $X$, then the edge $s \to t$ is labelled with probability $\prod_{\clauseC \in X} \lab{\clauseC} \cdot \prod_{\clauseC'\in C\setminus X} (1 - \lab{\clauseC'})$ --- recall that $\lab{\clauseC}$ is the probability labelling $\clauseC$. Also, $t$ has exactly $|X|$-many (clause-)children, each standing for a clause $\clauseC\in X$. If a child $u$ stands for $\clauseC = r:: A'\ot B_1,\dots, B_k$, then $u$ has $k$ (atom-)children, labelled with $B_1,\dots, B_k$ respectively.
\end{enumerate}
\end{defi}

Comparing distribution trees with stochastic derivation trees (Definition \ref{def:stochastictrees_ground}), one can observe the addition of another class of nodes, representing possible worlds. Moreover, the possible worlds associated with an atom-node must form a probability distribution --- as opposed to stochastic derivation trees, in which probabilities labelling parallel edges do not need to share any relationship. An example of the distribution tree associated with a goal is provided in the continuation of our leading example (Examples \ref{ex:ground} and \ref{example_base_and-or}).
\begin{exa}\label{ex:distrsemantics_ground}
In the context of Example \ref{ex:ground}, the distribution tree of $\sf{Hear\_alarm(M)}$ is depicted below, where we use grey shades to emphasise sets of edges expressing a probability distribution. Also, note the $\circ$s with no children, standing for empty worlds (namely worlds containing no clause).
\begin{equation*}
\hspace{-1cm}\includegraphics[height=6cm]{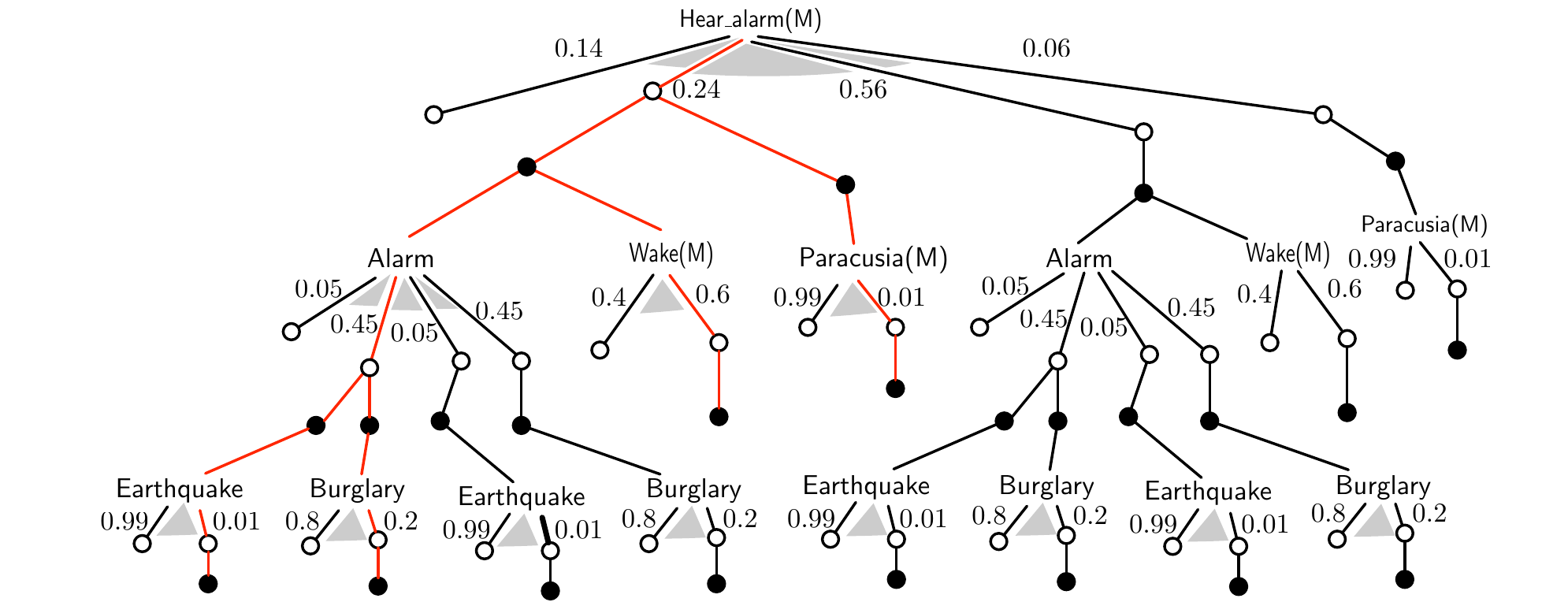}.
\end{equation*}
In the literature, the distribution semantics usually associates with a goal a single probability value~\eqref{eq:distrsemantics}, rather than a whole tree. However, given the distribution tree it is straightforward to compute such probability. The subtree highlighted in red above describes a refutation of $\sf{Hear\_alarm(M)}$ with probability 0.000001296 ( = the product of all the probabilities in the subtree). The sum of all the probabilities associated to such `proof' subtrees yields the usual distribution semantics~\eqref{eq:distrsemantics} --- the computation is shown in detail in Appendix~\ref{sec:algo_ground}.
\end{exa}

In the remainder of this section, we focus on the coalgebraic characterisation of distribution trees and the associated semantics map. Our strategy will be to introduce a novel coalgebra type $\G \powset_f$, such that distribution trees can be seen as elements of the cofree coalgebra. Then, we will provide a natural transformation $\F \To \G$, which can be used to transform stochastic derivation trees into distribution trees. Finally, composing the universal properties of these cofree coalgebras will yield the desired distribution semantics.



We begin with the definition of $\G$. This is simply the composite $\SubDist\powset_f \colon \sets\to \sets$, where $\SubDist$ is the \emph{sub-probability distribution} functor. Recall that $\SubDist$ maps $X$ to the set of sub-probability distributions with finite supports on $X$ (i.e., convex combinations of elements of $X$ whose sum is less or equal to $1$), and acts component-wise on functions.

\begin{rem} Note that we cannot work with full probabilities $\Dist$ here, since a goal may not match any clause. In such a case there is no world in which the   goal is provable and its probability in the program is $0$.
\end{rem}


The next step is to recover distribution trees as the elements of the $\G\powset_f$-cofree coalgebra on $\At$. This goes via a terminal sequence, similarly to the case of $\F\powset_f$ in the previous section. We will not go into full detail, but only mention that the terminal sequence for $\At\times \G\powset_f(-):\sets\to \sets$ is constructed as the one for $\At\times \F\powset_f(-):\sets\to \sets$ (Construction \ref{constr:terminalseq_ground}), with $\G$ replacing $\F$.

\begin{prop}
The terminal sequence of $\At\times \G \powset_f (-)$ converges at some limit ordinal $\chi$, and $(\lambda^{\chi+1}_\chi)^{-1}:Y_\chi \to \At\times \G \powset_f Y_\chi$ is the final $\At\times \G \powset_f$ coalgebra.
\end{prop}
\begin{proof}
As for Proposition~\ref{cor:terminalseq_ground_ok}, by \cite[Cor. 3.3]{worrell1999terminal} it suffices to show that $\G$ is accessible and preserves monos. Both are simple exercises; in particular, see \cite{BartelsSokolova} for accessibility of $\SubDist$.
\end{proof}

The association of distribution trees with elements of $Y_\chi$ is suggested by the type $\At \times \G \powset_f$. Indeed, $\At \times \SubDist$ is the layer of atom-nodes, labelled with elements of $\At$ and with outgoing edges forming a sub-probability distribution; the first $\powset_f$ is the layer of world-nodes; the second $\powset_f$ is the layer of clause-nodes. The coalgebra map $Y_\chi \to \At\times \G \powset_f Y_\chi$ associates a goal to subtrees of its distribution trees, analogously to the coalgebra structure on stochastic derivation trees in the previous section.

The last ingredient we need is a translation of stochastic derivation trees into distribution trees. We formalise this as a natural transformation $\posw: \F\To \G$. The naturality of $\posw$ can be checked with a simple calculation.
\begin{defi}\label{def:possibleworldnat}
The `\textbf{p}ossible \textbf{w}orlds' natural transformation $\posw \colon  \F \!\Rightarrow \!\G$ is defined~by $\posw_X \colon \phi \mapsto \sum_{Y\subseteq \supp(\phi)} r_Y Y $, where each $r_Y = \prod_{y\in Y}\phi(y) \cdot \prod_{y'\in \supp(\phi) \setminus Y} (1 - \phi(y'))$. In particular, when $\supp(\phi)$ is empty, $\posw_X(\phi)$ is the empty sub-distribution $\emptyDist$.
\end{defi}

Now we have all the ingredients to characterise the distribution semantics coalgebraically, as the morphism $\intd{-}_p \colon \At \to Y_{\chi}$ defined by the following diagram, which maps $A \in \At$ to its distribution tree in $p$.

\begin{equation}\label{const:dist_tree_groundcase}
\vcenter{
\xymatrix@R=25pt@C=50pt{
\At \ar@{-->}@/^25pt/[rrr]^{\intd{-}_p} \ar@{-->}[r]^{\intp{-}_p} \ar[d]^{<\id_\At,p>} & X_\gamma  \ar@{-->}[rr]^{\poswtr} \ar[d]^{\cong} && Y_\chi \ar[dd]^{\cong} \\
\At \times \F \powset_f \At \ar[r]^{\id_\At \times \F \powset_f(\intp{-}_p)} \ar[d]^{\id_\At\times \posw_{\powset_f(\At)} }  & \At\times \F \powset_f X_\gamma \ar[d]^{\id_\At \times \posw_{\powset_f (X_\gamma)} } && \\
\At \times \G \powset_f \At \ar[r]^{ \id_\At \times \G \powset_f (\intp{-}_p) }  & \At\times \G \powset_f X_\gamma \ar[rr]^{\id_\At \times \G \powset_f (\poswtr)} && \At\times \G \powset_f Y_\chi
}
}
\end{equation}
Note the use of $\posw$ to extend probabilistic logic programs and stochastic derivation trees to the same coalgebra type as distribution trees. Then the distribution semantics $\intd{-}_p$ is uniquely defined by the universal property of the final $\At \times \G \powset_f$-coalgebra. By uniqueness, it can also be computed as the composite $\poswtr \circ \intp{-}_p$, that is, first one derives the semantics $\intp{-}_p$, then applies the translation $\posw$ to each level of the resulting stochastic derivation tree, in order to turn it into a distribution tree.

\section{General Case}\label{sec:general_case}
We now generalise our coalgebraic treatment to arbitrary probabilistic logic programs and goals, possibly including variables. The section has a similar structure as the one devoted to the ground case. First, in Subsection \ref{sec:coalg_gen_LP}, we give a coalgebraic representation for general $\PLP$, and equip it with a final coalgebra semantics in terms of stochastic derivation trees (Subsection \ref{sec:coalg_general}). Next, in Subsection \ref{sec:gen_dist_semantics}, we study the coalgebraic representation of the distribution semantics. We begin by introducing our leading example --- an extension of Example~\ref{ex:ground}.

\begin{exa}\label{ex:general}
We tweak the ground program of Example~\ref{ex:ground}. Now it is not just Mary that may hear the alarm, but also her neighbours. There is a small probability that the alarm rings because someone passes too close to Mary's house. However, we can only estimate the possibility of paracusia and being awake for Mary, not the neighbours. The revised program, which by abuse of notation we also call $\progalarm$, is based on an extension of the language in Example \ref{ex:ground}: we add a new $1$-ary function symbol $\sf{Neigh}^1$ to the signature $\Sigal$, and a new $1$-ary predicate $\mathsf{PassBy(-)}$ to the alphabet. Note the appearance of a variable $x$.
\begin{center}
\begin{tabular}{rll|rll}
$0.01 :: $ & $\sf{Earthquake}$ & $\ot $ &
$0.5 :: $ & $\sf{Alarm}$ & $\ot \sf{Earthquake}$ \\
$0.2 :: $ & $\sf{Burglary}$ & $\ot$ &
$0.9 :: $ & $\sf{Alarm}$ & $\ot \sf{Burglary}$ \\
$0.6 :: $ & $\sf{Wake(Mary)}$ & $\ot$ &
$0.1 :: $ & $\sf{Alarm}$ & $\ot \sf{PassBy(x)}$ \\
$0.01:: $ & $\sf{Paracusia(Mary)}$ & $\ot $ &
$0.3 :: $ & $\sf{Hear\_alarm(x)}$ & $\ot \sf{Paracusia(x)}$ \\
$0.8 ::$ & $\sf{Wake(Neigh(x))}$ & $\ot \sf{Wake(x)}$ &
$0.8 :: $ & $\sf{Hear\_alarm(x)}$ & $\ot \sf{Alarm},~ \sf{Wake(x)}$

\end{tabular}
\end{center}
\end{exa}

As we want to maintain our approach a direct generalisation of the coalgebraic semantics~\cite{BonchiZ15} for pure logic programs, the derivation semantics $\intp{-}$ for general $\PLP$ will represent resolution by \emph{unification}. This means that, at each step of the computation, given a goal $A$, one seeks substitutions $\theta,\tau$ such that $A \theta = H \tau$ for the head $H$ of some clause in the program. As a roadmap, we anticipate the way this computation is represented in terms of stochastic derivation trees (Definition \ref{def:stoch_saturate_tree} below), with a continuation of our leading example.
\begin{exa}\label{ex:general_coalgsem}
In the context of Example~\ref{ex:general}, the tree for $\intp{\mathsf{Hear\_alarm(x)}}_{\progalarm}$ is (partially) depicted below. Compared to the ground case (Example \ref{example_base_and-or}), now substitutions applied on the goal side appear explicitly as labels. Moreover, note that for each atom $A\theta$ and clause $\clauseC$, there might be multiple substitutions $\tau$ such that $A\theta = \head{\clauseC}\tau$. This is presented in the following tree \eqref{eq:exampletreegeneral} by the $\blacklozenge$-labelled nodes, which are children of the clause-nodes (labelled with $\bullet$). We abbreviate $\mathsf{Neigh}$ as $\mathsf{N}$ and $\mathsf{Mary}$ as $\mathsf{M}$.
\begin{equation}\label{eq:exampletreegeneral}
\raisebox{-80pt}{\hbox{\includegraphics[height=6cm]{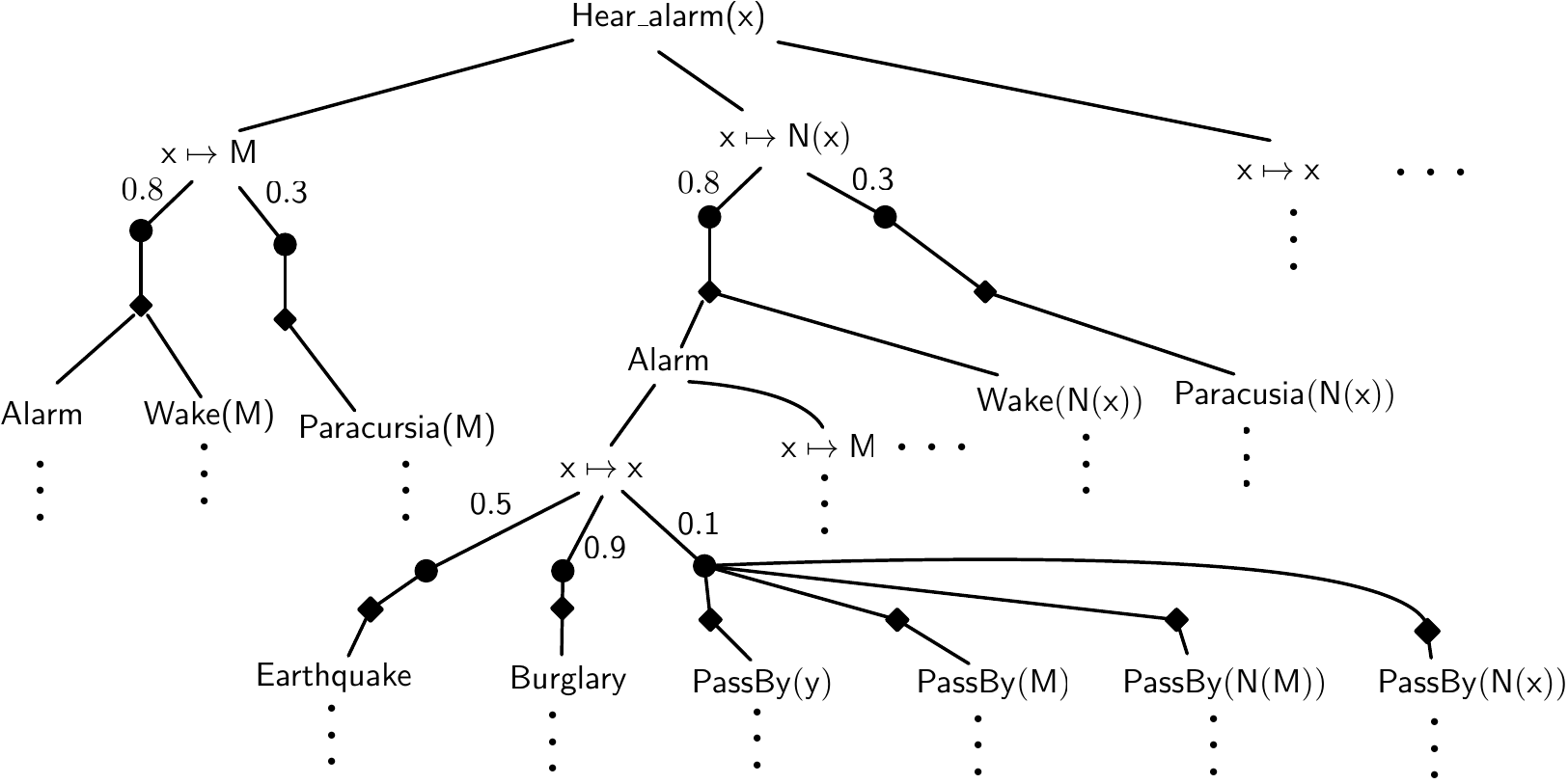}}}
\end{equation}
\end{exa}
Resolution by unification as above will be implemented in two stages. The first step is devising a map for term-matching. Assuming that the substitution instance $A \theta$ of a goal $A$ is already given, we define $p$ performing term-matching of $A \theta$ in a given program $\Prog$:

\begin{equation}\label{eq:termmatchingintuition}
\hspace{-.3cm}	\begin{array}{ccl}
 	p(A \theta) \colon & \{ B_1\tau_i,\dots, B_k\tau_i \}_{i \in I \subseteq \natN} & \mapsto  \
 	\begin{cases}
 		r \quad (r:: H \ot B_1,\dots, B_k) \in \Prog \text{ and } \\
 		\quad \quad I \text{ contains all } i \text{ s.t. }A \theta = H \tau_i \\
 	   0 \quad \text{otherwise.}
	\end{cases}
\end{array}
	\end{equation}
Intuitively, one application of such map is represented in a tree structure as Example \ref{ex:general_coalgsem} by the first two layers of the subtree rooted at $\theta$. The reason why the domain of $p(A)$ is a \emph{countable} set $\{ B_1\tau_i,\dots, B_k\tau_i \}_{i \in I \subseteq \natN}$ of instances of the same body $B_1,\dots,B_k$ is that the same clause may match a goal with countably many different substitutions $\tau_i$. For example in the bottom part of \eqref{eq:exampletreegeneral} there are countably infinite substitutions $\tau_i$ matching the head of $\mathsf{Alarm} \ot \mathsf{PassBy(x)}$ to the goal $\mathsf{Alarm}$, substituting $x$ with $\mathsf{Mary}, \mathsf{Neigh(Mary)}, \mathsf{Neigh(x)}, \dots$. This will be reflected in the coalgebraic representation of $\PLP$ (see \eqref{eq:TMcoalg} below) by the use of the countable powerset functor $\powset_c$.

In order to model arbitrary unification, the second step is considering all substitutions $\theta$ on the goal $A$ such that a term-matcher for $A \theta$ exists. There is an elegant categorical construction \cite{BonchiZ15} packing together these two steps into a single coalgebra map. We will start with this in the next session.

\subsection{Coalgebraic Representation of $\PLP$}
\label{sec:coalg_gen_LP}

Towards a categorification of general $\PLP$, the first concern is to account for the presence of variables in atoms. This is standardly done by letting the space of atoms on an alphabet $\alphabet$ be a presheaf $\At \colon \LawOp \to \sets$ rather than a set. Here the index category $\LawOp$ is the opposite \textit{Lawvere Theory} of $\Sigma$ (see Section~\ref{sec:preliminaries}). For each $n\in \ob(\LawOp)$, $\At(n)$ is defined as the set of $\alphabet$-atoms in context $n$. Given an $n$-tuple $\theta = \lr{t_1,\dots, t_n} \in \LawOp[n, m]$ of $\Sigma$-terms in context $m$, $\At(\theta) \colon \At(n)\to \At(m)$ is defined by substitution, namely $\At(\theta)(A) = A\theta$, for any $A\in \At(n)$.

As observed in \cite{komendantskaya2011coalgebraic} for pure logic programs, if we naively try to model our specification \eqref{eq:termmatchingintuition} for $p$ as a coalgebra on $\At$, we run into problems: indeed $p$ is not a natural transformation, thus not a morphism between presheaves. Intuitively, this is because the existence of a term-matching for a substitution instance $A\theta$ of $A$ does not necessarily imply the existence of a term-matching for $A$ itself. For pure logic programs, this problem can be solved in at least two ways. First, \cite{komendantskaya2011coalgebraic} relaxes naturality by changing the base category of presheaves from $\sets$ to $\poset$. We take here the second route, namely give a `saturated' coalgebraic treatment of $\PLP$, generalising the modelling of pure logic programs proposed in \cite{BonchiZ15}. This approach has the advantage of letting us work with $\sets$-based presheaves, and be still able to recover term-matching via a `desaturation' operation --- see \cite{BonchiZ15} and Appendix \ref{app:general_computation}.

\smallskip
 \noindent\textbf{The Saturation Adjunction.} To this aim, we briefly recall the saturated approach from \cite{BonchiZ15}. The central piece is the adjunction $\U \dashv \K$ on presheaf categories, as on the left below.
\begin{equation}\label{eq:saturationadju}
\xymatrix{
\sets^{\LawOp} \ar@/^1pc/[rr]^{\U} & \bot & \sets^{|\LawOp|} \ar@/^1pc/[ll]^{\K}
}
\qquad \qquad \qquad \qquad
\vcenter{
\xymatrix{
|\LawOp| \ar@{^{(}->}[r]^{\iota} \ar[d]_{\func{F}} & \LawOp \ar[ld]^{\K(\func{F})} \\
\sets  &
}
}
\end{equation}
Here $|\LawOp|$ is the discretisation of $\LawOp$, i.e. all the arrows but the identities are dropped. The left adjoint $\U$ is the forgetful functor, given by precomposition with the obvious inclusion $\iota \colon |\LawOp|\to \LawTh^\op_\Sigma$. $\U$ has a right adjoint $\K = \mathsf{Ran}\iota \colon \sets^{|\LawOp|} \to \sets^{\LawTh^\op_\Sigma}$, which sends every presheaf $\func{F} \colon |\LawTh^\op_\Sigma|\to \sets$ to its \textit{right Kan extension} along $\iota$, as in the rightmost diagram in~\eqref{eq:saturationadju}. The definition of $\K$ can be computed~\cite{maclane1971} as follows:
\begin{itemize}
	\item On objects $\func{F} \in \ob(\sets^{\LawOp})$, the presheaf $\K(\func{F}) \colon \LawOp \to \sets$ is defined by letting $\K(\func{F})(n)$ be the product $\prod_{\theta\in \LawOp[n,m]} \func{F}(m)$, where $m$ ranges over $\ob(\LawOp)$. Intuitively, every element in $\K(\func{F})(n)$ is a tuple with index set $\bigcup_{m\in \ob(\LawOp)}\LawOp[n,m]$, and its component at index $\theta \colon n\to m$ must be an element in $\func{F}(m)$. We follow the convention of \cite{BonchiZ15} and write $\dot{x}, \dot{y}, \dots$ for such tuples, and $\dot{x}(\theta)$ for the component of $\dot{x}$ at index $\theta$.

	With this convention, given an arrow $\sigma\in \LawOp[n,n']$, $\K (\func{F})(\sigma)$ is defined by pointwise substitution as the mapping of the tuple $\dot{x}$ to the tuple $\lr{\dot{x}(\theta\after \sigma)}_{\theta \colon n'\to m}$.
	\item On arrows, given a morphism $\alpha \colon \func{F}\to \func{G}$ in $\sets^{|\LawOp|}$, $\K(\alpha)$ is a natural transformation $\K(\func{F}) \to \K(\func{G})$ defined pointwisely
as $\K(\alpha)(n) \colon \dot{x}\mapsto \lr{\alpha_m(\dot{x}(\theta))}_{\theta \colon n\to m}$.
\end{itemize}
It is also useful to record the unit $\eta \colon 1\to \K\U$ of the adjunction $\U\dashv \K$. Given a presheaf $\func{F} \colon \LawTh^\op_\Sigma\to \sets$, $\eta_\func{F} \colon \func{F}\to \K\U\func{F}$ is a natural transformation defined by $ \eta_\func{F}(n) \colon x \mapsto \lr{\func{F}(\theta)(x)}_{\theta \colon n\to m}$.

 \noindent\textbf{Saturation in $\PLP$.} We now construct the coalgebra structure on the presheaf $\At$ modelling $\PLP$. First, we are now able to represent $p$ in \eqref{eq:termmatchingintuition} as a coalgebra map. The aforementioned naturality issue is solved by defining it as a morphism in $\sets^{|\LawOp|}$ rather than in $\sets^{\LawOp}$, thus making naturality trivial. The coalgebra $p$ will have the following type
\begin{equation}\label{eq:TMcoalg}
	p \colon \U\At\to \liftF \lift{\powset_c} \lift{\powset_f} \U\At
\end{equation}
 where $\lift{(\cdot)}$ is the obvious extension of $\sets$-endofunctors to $\sets^{|\LawOp|}$-endofunctors, defined by functor precomposition. With respect to the ground case, note the insertion of $\lift{\powset_c}$, the lifting of the \emph{countable} powerset functor, in order to account for the countably many instances of a clause that may match the given goal (\emph{cf.} the discussion below \eqref{eq:termmatchingintuition}). 
 
 \begin{exa}
Our program $\progalarm$ (Example \ref{ex:general}) is based on $\Atalarm \colon \LawTh^\op_{\Sigal} \to \sets$. Some of its values are $\Atalarm(0) = \{ \sf{Mary}, \sf{Neigh(Mary)}, \sf{Neigh(Neigh(Mary))}, \dots \}$ and $\Atalarm(1) = \{ \mathsf{x}, \sf{Mary}, \sf{Neigh(x)}, \sf{Neigh(Mary)}, \dots \}$. Part of the coalgebra $\palarm$ modelling the program $\progalarm$ is as follows (\emph{cf.} the tree \eqref{eq:exampletreegeneral}). For space reason, we abbreviate $\sf{Mary}$ as $\sf{M}$ in the second equation.
\begin{align*}
(\palarm)_0 (\sf{Hear\_alarm(Mary)}) & = 0.8 \{ \{\sf{Alarm}, \sf{Wake(Mary)}\} \} + 0.3\{ \{\mathsf{Parasusia(Mary)} \} \} \\
(\palarm)_1 (\mathsf{Alarm}) & = 0.5 \{ \{\mathsf{Earthquake}\} \} + 0.9 \{ \{\mathsf{Burglary}\} \} \\
& ~~~~ + 0.1 \{ \{\mathsf{PassBy(M)}\}, \{ \mathsf{PassBy(Neigh(M))} \}, \{ \mathsf{PassBy(Neigh(x))} \}, \dots  \}
\end{align*}
\end{exa}

 The universal property of the adjunction \eqref{eq:saturationadju} gives a canonical `lifting' of $p$ to a $\K \liftF \lift{\powset_c} \lift{\powset_f} \U$-coalgebra $p^\sharp$ on $\At$, performing unification rather than just term-matching:
\begin{equation}\label{eq:Unificationcoalg}
p^\sharp \coloneqq \At \xto{\eta_\At} \K\U \At \xto{\K p} \K \liftF \lift{\powset_c} \lift{\powset_f} \U \At
\end{equation}
where $\eta$ is the unit of the adjunction, as defined above. Spelling it out, $p^\sharp$ is the mapping
\[p^\sharp_n \ \colon \ A\in \At(n) \mapsto \lr{p_m (A\theta)}_{\theta \colon n\to m}
\]
Intuitively, $p^\sharp_n$ retrieves all the unifiers $\lr{\theta,
\tau}$ of $A$ and head $H$ in $\Prog$: first, we have $A\theta\in \At(m)$ as a component of the saturation of $A$ by $\eta_{\At}$; then we term-match $H$ with $A\theta$ by $\K p_m$.

\begin{rem}\label{rmk:context}
Note that the parameter $n \in \ob(\LawOp)$ in the natural transformation $p^{\sharp}$ fixes the pool $\{x_1,\dots,x_n\}$ of variables appearing in the atoms (and relative substitutions) that are considered in the computation. Analogously to the case of pure logic programs~\cite{komendantskaya2011coalgebraic, BonchiZ15}, it is intended that such $n$ can always be chosen `big enough' so that all the relevant substitution instances of the current goal and clauses in the program are covered --- note the variables occurring therein always form a \emph{finite} set, included in $\{x_1, \dots, x_m\}$ for some $m\in \natN$.
\end{rem}


\subsection{Derivation Semantics}\label{sec:coalg_general}

Once we have identified our coalgebra type, the construction leading to the derivation semantics $\intp{-}_{p^\sharp}$ for general $\PLP$ is completely analogous to the ground case. One can define the cofree coalgebra for $\K \liftF \lift{\powset_c} \lift{\powset_f} \U (-)$ by terminal sequence, similarly to Construction \ref{constr:terminalseq_ground}. For simplicity, henceforth we denote the functor $\K \liftF \lift{\powset_c} \lift{\powset_f} \U(-)$ by $\genPLP$. 
\begin{construction}\label{constr:terminalseq_general}
The terminal sequence for $\At\times \genPLP(-): \sets^{\LawOp} \to \sets^{\LawOp}$ consists of a sequence of objects $X_\alpha$ and morphisms $\delta^\beta_\alpha \colon X_\beta\to X_\alpha$, for $\alpha<\beta \in \Ord$, defined analogously to Construction \ref{constr:terminalseq_ground}, with $p^\sharp$ and $\genPLP$ replacing $p$ and $\F \powset_f$.
\end{construction}
This terminal sequence converges by the following proposition.
\begin{prop}\label{lemma:gen.acc}
$\S$ is accessible, and preserves monomorphisms.
\end{prop}
\begin{proof}
Since both properties are preserved by composition, it suffices to show that they hold for all the component functors. For $\liftF$, $\lift{\powset_c}$ and $\lift{\powset_f}$, they follow from accessibility and mono-preservation of $\F$, $\powset_c$ and $\powset_f$ (see Proposition~\ref{cor:terminalseq_ground_ok}), as (co)limits in presheaf categories are computed pointwise. For $\K$ and $\U$, these properties are proven in \cite{BonchiZ15}.
\end{proof}
Therefore the terminal sequence for $\At\times \genPLP(-)$ converges at some limit ordinal, say $\gamma$, yielding the final $\At\times \genPLP(-)$-coalgebra $X_{\gamma} \xrightarrow{\iso} \At\times \genPLP(X_\gamma)$. The derivation semantics is then defined $\intp{-}_{p^\sharp} \colon \At\to X_\gamma$ by universal property as in the diagram below.
\begin{equation}\label{gen_final_sem}
\vcenter{
\xymatrix{
\At \ar@{.>}[rr]^-{\intp{-}_{p^\sharp}} \ar[d]_{<\id_\At, p^\sharp>} && X_\gamma \ar[d]^{\iso} \\
\At\times \genPLP(\At) \ar[rr]_-{\id_\At \times \intp{-}_{p^\sharp}} && \At\times \genPLP(X_\gamma)
}
}
\end{equation}
A careful inspection of the terminal sequence constructing $X_{\gamma}$ allows to infer a representation of its elements as trees, among which we have those representing computations by unification of goals in a $\PLP$ program. We call these \emph{stochastic saturated derivation trees}, as they extend the derivation trees of Definition \ref{def:stochastictrees_ground} and are the probabilistic variant of saturated and-or trees in \cite{BonchiZ15}. Using \eqref{gen_final_sem} one can easily verify that $\intp{A}_{p^\sharp}$ is indeed the stochastic saturated derivation tree for a given goal $A$. Example \ref{ex:general_coalgsem} provides a pictorial representation of one such tree.

\begin{defi}[Stochastic saturated derivation trees]\label{def:stoch_saturate_tree}
Given a probabilistic logic program $\Prog$, a natural number $n$ and an atom $A\in \At(n)$. The \emph{stochastic saturated derivation tree} for $A$ in $\Prog$ is the possibly infinite tree $\mathcal{T}$ satisfying the following properties:
\begin{enumerate}
	\item There are four kinds of nodes: atom-node (labelled with an atom), substitution-node (labelled with a substitution), clause-node (labelled with $\clauseCnode$), instance-node (labelled with $\blacklozenge$), appearing alternatively in depth in this order. The root is an atom-node with label~$A$.
	\item Each edge between a substitution node and its clause-node children is labelled with some $r\in [0,1]$.
	\item Suppose an atom-node $s$ is labelled with $A'\in \At(n')$. For every substitution $\theta \colon n'\to m'$, $s$ has exactly one (substitution-node) child $t$ labelled with $\theta$. For every clause $r :: H\ot B_1,\dots, B_k$ in $\Prog$ such that $H$ matches $A'\theta$ (via some substitution), $t$ has exactly one (clause-)child $u$, and edge $t\to u$ is labelled with $r$. Then for every substitution $\tau$ such that $A'\theta = H\tau$ and $B_1\tau, \dots, B_k\tau \in \At(m')$, $u$ has exactly one (instance-)child $v$. Also $v$ has exactly $|\{ B_1, \dots, B_k\}|\tau$-many (atom-)children, each labelled with one element in $\{ B_1, \dots, B_k \}\tau$.
\end{enumerate}
\end{defi}

\subsection{Distribution Semantics} \label{sec:gen_dist_semantics}
In this section we conclude by giving a coalgebraic perspective on the distribution semantics $\intd{-}$ for general $\PLP$. 
Mimicking the ground case (Section~\ref{sec:distrsem_groundcase}), this will be presented as an extension of the derivation semantics, via a `possible worlds' natural transformation. Also in the general case, we want to guarantee that a single probability value is computable for a given goal $A$ from the corresponding tree $\intd{A}$ in the final coalgebra --- whenever this probability is also computable in the `traditional' way (see \eqref{eq:distrsemantics}) of giving distribution semantics to $\PLP$. In this respect, the presence of variables and substitutions poses additional challenges, for which we refer to Appendices~\ref{sec:algo_ground} and~\ref{app:general_computation}. In a nutshell, the issue is that the distribution semantics counts the existence of a clause in the program rather than the number of times a clause is used in the computation. This means that every clause is counted at most once, independently from how many times that clause is used again in the computation. Note that our tree representation, as in the ground case, presents the resolution (thus computation aspect) of the program. So in our tree representation we need to give enough information to determine which clause is used at each step of the computation, so that a second use can be easily detected. 
However, neither our saturated derivation trees, nor a `naive' extension of them to distribution trees, carry such information: what appears in there is only the instantiated heads and bodies, but in general one cannot retrieve $A$ from a substitution $\theta$ and the instantiation $A\theta$. This is best illustrated via a simple example.
\begin{exa}
Consider the following program, based on the signature $\Sigma = \{ a^0 \}$ and two $1$-ary predicates $P$, $Q$. It consists of two clauses:
\begin{equation*}
\begin{tabular}{rl|rl}
$0.5 ::$ & $P(x_1)\ot Q(x_1)$  \quad \quad   & \quad \quad  $0.5 ::$ & $P(x_1)\ot Q(x_2)$
\end{tabular}
\end{equation*}
The goal $P(a)$ matches the head of both clauses. However, given the sole information of the next goal being $Q(a)$, it is impossible to say whether the first clause has been used, instantiated with $x_1 \mapsto a$, or the second clause has been used, instantiated with $x_1 \mapsto a, x_2 \mapsto a$.
\end{exa}

This observation motivates, as intermediate step towards the distribution semantics, the addition of labels to clause-nodes in derivation trees, in order to make explicit which clause is being applied. From the coalgebraic viewpoint, this just amounts to an extension of the type of the term-matching coalgebra:
\begin{equation}\label{eq:gen-PLP-dist-term-match}
\pclauses \colon \U\At \to \liftF (\lift{\powset_c} \lift{\powset_f} \U \At \times (\U \At \times \lift{\powset_f} \U \At) )
\end{equation}
Note the insertion of $(-) \times (\U \At \times \lift{\powset_f} \U \At)$, which allows us to indicate at each step the head ($\U\At$) and the body ($ \lift{\powset_f} \U \At$) of the clause being used, its probability label being already given by $\liftF$. More formally, for any $n$ and atom $A\in \At(n)$, we define\footnote{As noted in Remark~\ref{rmk:context}, instantiating $\pclauses$ to some $n \in \ob(\LawOp)$ fixes a variable context $\{x_1,\dots,x_n\}$ both for the goal and the clause labels. In practice, because the set of clauses is always finite, it suffices to chose $n$ `big enough' so that the variables appearing in the clauses are included in $\{x_1,\dots,x_n\}$.}
\begin{align*}
\pclauses_n(A) \colon  \lr{ \{B_1\tau_i, \dots, B_k\tau_i\}_{i\in \mathcal{I}\subseteq \natN} , \lr{H, \{B_1,\dots, B_k\}} } 
\mapsto \begin{cases}
r \quad (r \! :: \! H \! \ot\! B_1,\dots, B_k) \in \Prog \\
\quad\quad \text{and } H\tau_i = A\\
0 \quad \text{otherwise}
\end{cases}
\end{align*}
As in the case of $p$ in \eqref{eq:TMcoalg}, we can move from term-matching to unification by using the universal property of the adjunction $\U\dashv \K$, yielding $\pclauses^{\, \sharp} \colon \At \to  \K \liftF (\lift{\powset_c} \lift{\powset_f} \U \At \times (\U \At \times \lift{\powset_f} \U \At) )$. For simplicity henceforth we denote the functor $\K \liftF (\lift{\powset_c} \lift{\powset_f} \U (-) \times (\U \At \times \lift{\powset_f} \U \At) )$ by $\genPLPclause$.

We are now able to conclude our characterisation of the distribution semantics. The `possible worlds' transformation $\posw \colon \F \To \G$ (Definition \ref{def:possibleworldnat}) yields a natural transformation $\lift{\posw} \colon \liftF \to \liftG$, defined pointwise by $\posw$. We can use $\lift{\posw}$ to translate $\genPLPclause$ into the functor $\K \liftG (\lift{\powset_c} \lift{\powset_f} \U (-) \times (\clauseCfunc))$, abbreviated as $\genDist$, which is going to give the type of saturated distribution trees for general $\PLP$ programs.

As a simple extension of the developments in Section~\ref{sec:coalg_general}, we can construct the cofree $\genPLPclause$-coalgebra $\Phi \xrightarrow{\iso} \At \times \genPLPclause(\Phi)$ via a terminal sequence. Similarly, one can obtain the cofree $\genDist$-coalgebra $\Psi \xrightarrow{\iso} \At \times \genDist(\Psi)$. By the universal property of $\Psi$, all these ingredients get together in the definition of the distribution semantics $\intd{-}_{\pclauses^{\, \sharp}}$ for arbitrary $\PLP$ programs $\pclauses^{\, \sharp}$
$$
\xymatrix@R=25pt@C=50pt{
\At \ar@{-->}@/^25pt/[rrr]^{\intd{-}_{\pclauses^{\, \sharp}}} \ar@{-->}[r]^{!_{\Phi}} \ar[d]^{<\id_\At,\pclauses^{\, \sharp}>} & \Phi  \ar@{-->}[rr]^{\poswtr_{\Psi}} \ar[d]^{\cong} && \Psi \ar[dd]^{\cong} \\
\At \times \genPLPclause \At \ar[r]^{\id_\At \times \genPLPclause(!_{\Phi})} \ar[d]^{\id_\At\times \K\lift{\posw}}  & \At\times \genPLPclause \Phi \ar[d]^{\id_\At \times \K\lift{\posw}} && \\
\At \times \genDist \At \ar[r]^{ \id_\At \times \genDist (!_{\Phi}) }  & \At\times \genDist \Phi \ar[rr]^{\id_\At \times \genDist (\poswtr_{\Psi})} && \At\times \genDist \Psi
}
$$
where $!_{\Phi}$ and $\poswtr_{\Psi}$ are given by the evident universal properties, and show the role of the cofree $\genPLPclause$-coalgebra $\Phi$ as an intermediate step. The layered construction of final coalgebras $\Psi$ and $\Phi$, together with the above characterisation of $\intd{-}_{\pclauses^{\, \sharp}}$, allow to conclude that the distribution semantics for the program $\pclauses^{\, \sharp}$ maps a goal $A$ to its \emph{saturated distribution tree} $\intd{A}_{\pclauses^{\, \sharp}}$, as formally defined below.

\begin{defi}[Saturated distribution tree]\label{def:saturated-dist-tree}
The \emph{saturated distribution tree} for $A\in \At(n)$ in $\Prog$ is the possibly infinite $\mathcal{T}$ satisfying the following properties based on Definition \ref{def:stoch_saturate_tree}:
\begin{enumerate}
	\item There are five kinds of nodes: in addition to the atom-, substitution-, clause- and instance-nodes, there are world-nodes. The world-nodes are children of the substitution-nodes, and parents of the clause nodes. The root and the order of the rest of the nodes are the same as in Definition \ref{def:stoch_saturate_tree}, condition \textbf{1}. The clause-nodes are now labelled with clauses in $\Prog$.
	\item Suppose $s$ is an atom-node labelled with $A'\in \At(n')$, and $t$ is a substitution-child of $s$ labelled with $\theta \colon n'\to m$. Let $C$ be the set of all clauses $\clauseC$ such that $\head{\clauseC}$ matches $A'\theta$. Then $t$ has $2^{|C|}$ world-children, each representing a subset $X$ of $C$. If a world-child $u$ of $t$ represents subset $X$, then the edge $t \to u$ has probability label $\prod_{\clauseC\in X} \lab{\clauseC} \cdot \prod_{\clauseC'\in C\setminus X} (1- \lab{\clauseC'})$. Also $u$ has $|X|$ clause-children, one for each clause $\clauseC\in X$, labelled with the corresponding clause. The rest for clause-nodes and instance-nodes are the same as in Definition \ref{def:stoch_saturate_tree}, condition \textbf{3}.
\end{enumerate}

\begin{rem}
Note that, in principle, saturated distribution trees could be defined coalgebraically without the intermediate step of adding clause labels. This is to be expected: coalgebra typically captures the one-step, `local' behaviour of a system. On the other hand, as explained, the need for clause labels is dictated by a computational aspect involving the depth of distribution trees, that is, a `non-local' dimension of the system.
\end{rem}

\end{defi}
We conclude with the pictorial representation of the saturated distribution tree of a goal in our leading example.
\begin{exa}
In the context of Example \ref{ex:general}, the tree $\intd{\mathsf{Hear\_alarm(x)}}$ capturing the distribution semantics of $\mathsf{Hear\_alarm(x)}$ is (partially) depicted as follows. Note the presence of clauses labelling the clause-nodes. 
\begin{equation*}
\raisebox{-80pt}{\hbox{\includegraphics[height=6cm]{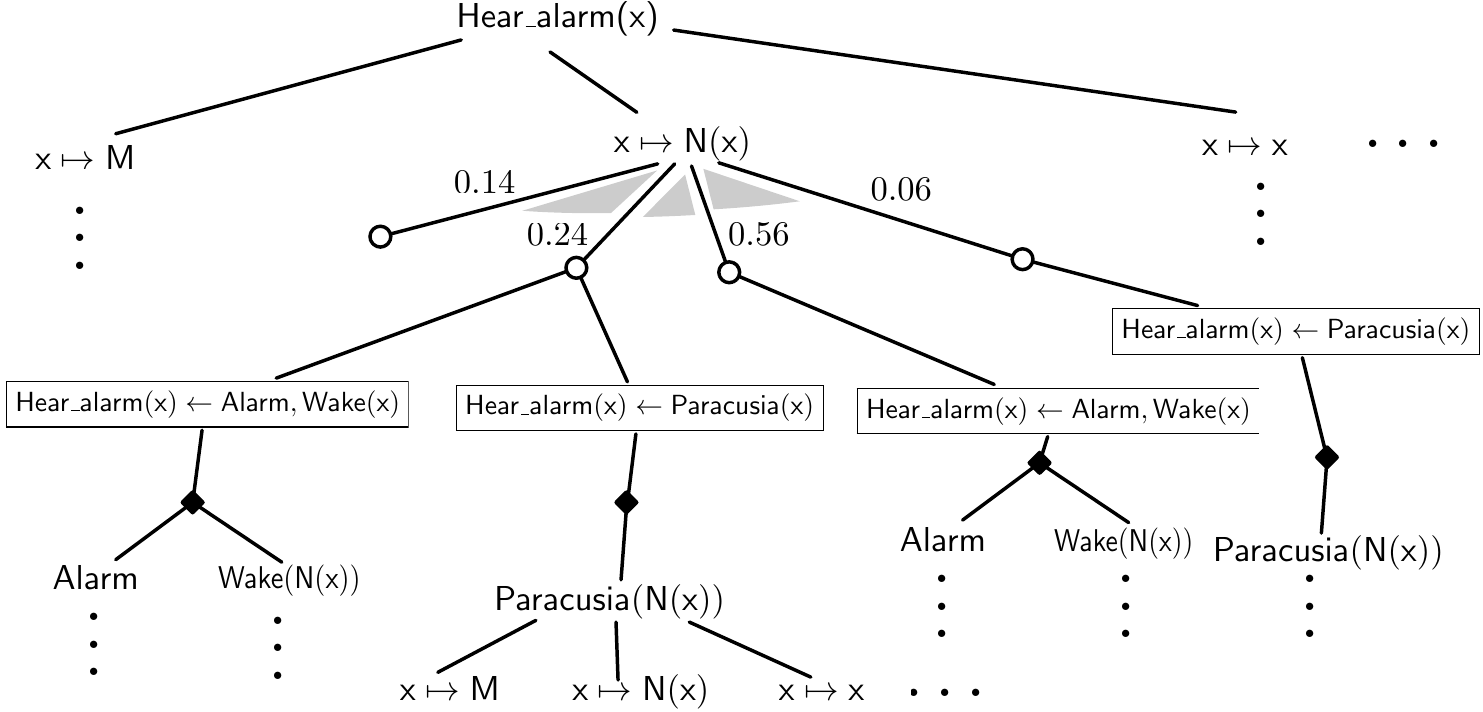}}}
\end{equation*}
\end{exa}


\section{Weighted Logic Programming}\label{sec:WLP}
Weighted logic programming ($\WLP$) is an abstract framework for describing dynamic programming algorithms in a number of fields, such as natural language processing \cite{eisner2007} and computational biology \cite{durbin1998}. $\WLP$ generalises logic programming by assigning a weight/score to each clause. Given a goal $A$, the central task then becomes computing the weight of $A$, based on the weights of the clauses appearing in the proof tree of $A$.

As $\WLP$ and $\PLP$ are seemingly close formalisms, it is natural to ask to what extent our coalgebraic perspective applies to $\WLP$. In this section we develop such perspective. By analogy with the development of coalgebraic $\PLP$, we first discuss the ground case (Subsection \ref{sec:WLP-groundcase}), and then move to the general case (Subsection \ref{sec:WLP-generalcase}). 

\medskip
We begin by recalling the basics of $\WLP$, referring to \cite{eisner2007} for further details. We fix a semiring $\semiK = \lr{K, \semiPlus, \semiTimes, \plusUnit, \prodUnit}$, where $\plusUnit$ is the $\semiPlus$-unit, and $\prodUnit$ is the $\semiTimes$-unit. A weighted logic program $\progW$ based on a logic program $\progL$ assigns to each clause $\clauseC$ in $\progL$ a weight label $c\in K$. In other words, a weighted logic program $\progW$ consists of finitely many clauses of the form:
\[
c :: H \ot B_1, \dots, B_k
\]
where $H, B_1, \dots, B_k$ are atoms and $c\in K$.\footnote{ It is typical in the literature (e.g. \cite{shay2010}) to only assign weights to facts, i.e. clauses with empty bodies, and leave the other clauses unlabelled. However, we can assume without loss of generality that all clauses are assigned a weight, as the weight semantics \eqref{eq:cal-weight} implicitly assigns weight $\prodUnit$ to unlabelled clause.}
\begin{exa}\label{ex:weight-short-path}
We use as running example the following weighted logic program $\progShortPath$, from \cite{shay2010}. $\progShortPath$ is based on the semiring $\semiK^{sp} = \lr{\nonNegReal \cup \{ +\infty\}, \min, +, +\infty, 0}$. Its language consists of constants $\sf{ a, b, c, d }$, a 1-ary predicate $\sf{initial(-)}$, and 2-ary predicates $\sf{edge}(-, -)$ and $\sf{reachable(-, -)}$.
\medskip

\begin{minipage}{0.5\textwidth}
\begin{tabular}{rllrll}
$0 ::$ & $\mathsf{initial(a)}$ & $\ot$ & \multicolumn{1}{|r}{$4::$} & $\mathsf{edge(a,c)}$ & $\ot$ \\
$20 ::$ & $\sf{edge(a,d)}$ & $\ot$ & \multicolumn{1}{|r}{$8::$} & $\sf{edge(b,b)}$ & $\ot$ \\
$9 ::$ & $\sf{edge(c,a)}$ & $\ot$ & \multicolumn{1}{|r}{$15 ::$} & $\mathsf{edge(c,d)}$ & $\ot$ \\
$6 ::$ & $\sf{edge(d, b)}$ & $\ot$ & \multicolumn{1}{|r}{$2 ::$} & $\sf{edge(d,d)}$ & $\ot$ \\
$6 ::$ & $\sf{edge(d,b)}$ & $\ot$ & \multicolumn{1}{|r}{$0 ::$} & $\sf{reachable(x)}$ & $\ot \sf{initial(x)}$ \\
$0 ::$ & $\mathsf{reachable(x)}$ & $\ot$ & \multicolumn{3}{l}{$\sf{reachable(y)}, \sf{edge(y, x)}$}
\end{tabular}
\end{minipage}
\noindent\begin{minipage}{0.5\textwidth}
\begin{equation*}
\xymatrix{
a \ar@/_/[d]|{\whitespacearound{4}} \ar@/^/[dr]|{\whitespacearound{20}} & b \ar@(l,u)[]^{8}  \\
c \ar@/_/[r]_{15} \ar@/_/[u]|{\whitespacearound{9}} & d \ar@/_/[l]|{\whitespacearound{16}} \ar@/_/[u]_{6} \ar@(rd, ru)[]_{2}
}
\end{equation*}
\end{minipage}
\medskip

\noindent $\progShortPath$ describes the single-source directed graph on the right. It expresses that a state $x$ is reachable if either $x$ itself is the initial state, or there is an edge from a reachable state to $x$. For simplicity, we will sometimes refer to the predicates appearing in $\progShortPath$ by their abbreviations, such as $\sf{init(-)}$, $\sf{reach(-, -)}$.
\end{exa}

As for the semantics, the central task for $\WLP$ is computing the \emph{weight} $\weight^{\scriptscriptstyle \progW}(A)$ (or simply $\weight(A)$) of an atom $A$ in $\progW$. This is an element of $\semiK$, defined as follows:

\begin{equation}\label{eq:cal-weight}
\weight^{\scriptscriptstyle \progW}(A) = \semiSum_{\substack{c :: A \ot B_1, \dots, B_k ~\text{is substitution} \\ \text{instance of some clause in}~\progW}} \left( c \semiTimes \semiProd_{B_i} \weight^{\scriptscriptstyle \progW}(B_i) \right)
\end{equation}
We will refer to \eqref{eq:cal-weight} as the \emph{weight semantics} for $\WLP$. Note that, when the body of a clause is empty, $\semiProd_{B_i} \weight(B_i)$ is an empty product and yields the unit $\prodUnit$. 
Also, the weight function $\weight$ is not necessarily well-defined for every goal and program. For example, consider the program consisting of a single clause $2 :: A \ot A$ with the semiring in Example \ref{ex:weight-short-path}. Then \eqref{eq:cal-weight} does not assign a value to $A$, because the recursive calculation does not terminate. 
\begin{exa}
In the program $\progShortPath$ from Example \ref{ex:weight-short-path}, the goal $\sf{reachable(d)}$ has weight $\weight(\sf{reachable(d)}) = 19$, and it is exactly the weight of the shortest path from $\sf{a}$ to $\sf{d}$.
\end{exa}

\begin{rem}\label{rem:PLP-is-not-WLP}
One might wonder whether $\PLP$ in Section \ref{sec:groundcase} and \ref{sec:general_case} is a special case of $\WLP$. The answer is negative, for two reasons. First, the underlying algebra structure for $\PLP$ is $\lr{[0,1], \por, \times, 0, 1}$, which is not a semiring (distributivity fails). Second, the distribution semantics in \eqref{eq:distrsemantics} is not a special case of the weight semantics in \eqref{eq:cal-weight}.
\end{rem}

\subsection{Coalgebraic Semantics of Ground $\WLP$}\label{sec:WLP-groundcase}
We will present the coalgebraic semantics for ground $\WLP$ analogous to that for $\PLP$ in Section \ref{sec:groundcase}: we first give a coalgebraic presentation of $\WLP$; then we induce a derivation semantics for it. The computation of the weight semantics from the derivation semantics is in Appendix \ref{app:compute-weight}.

A ground $\WLP$ program $\progW$ can be represented as a coalgebra of the type $\multiFunc \fPowset \colon \sets\to \sets$, where $\multiFunc$ is the multiset functor corresponding to the underlying semiring $\semiK = \lr{K, \semiPlus, \semiTimes, \plusUnit, \prodUnit}$. On objects, $\multiFunc \colon \sets\to \sets$ maps a set $A$ to $\{ \phi \colon A\to \semiK \mid \supp(\phi) ~\text{is finite} \}$, where $\supp(\phi) = \{ a\in A \mid \phi(a) \neq \plusUnit \}$. For each $a\in A$, we call $\phi(a)$ the weight of $a$ (under $\phi$). We adopt the standard notation $\sum_{i=1}^{k} c_i a_i$ for the function $\phi \in \multiFunc(A)$ with support $\supp(\phi) = \{ a_1, \dots, a_k \}$ and weights $\phi(a_i) = c_i$. We also use $\varnothing$ for the multiset with empty support. On morphisms, $\multiFunc(h \colon A\to B)$ maps $\sum_{i=1}^{k} c_i a_i$ to $\sum_{i=1}^{k} c_i h(a_i)$. 

Fix a ground $\WLP$ program $\progW$ based on a semiring $\semiK$ and a set of atoms $\At$. The definition of $\progW$ can be encoded as an $\multiFunc \fPowset$-coalgebra $w \colon \At\to \multiFunc \fPowset (\At)$ as follows: for each $A\in \At$,
\begin{equation*}\label{eq:ground-WLP-coalg-def}
\begin{array}{ccl}
 	w(A)  \colon  & \powset_f(\At) & \to \ K \\[0.5em]
 	& \{B_1,\dots, B_n\} & \mapsto  \
 	\begin{cases}
 		c \quad \text{if } c :: A \ot B_1,\dots, B_n \text{ is a clause in } \progW \\
 	   \plusUnit \quad \text{otherwise.}
	\end{cases}
\end{array}
\end{equation*}
\begin{exa}\label{ex:ground-WLP}
Consider the following ground program $\progGroundSP$ obtained as a simple ground version of $\progShortPath$. The underlying semiring is still $\semiK^{sp} = \lr{\nonNegReal \cup \{ +\infty\}, \min, +, +\infty, 0}$.

\medskip
\begin{center}
\begin{tabular}{rll | rll}
$0::$ & $\sf{initial(a)}$ & $\ot$ & $0::$ & $\sf{reachable(d)}$ & $\ot \sf{initial(d)}$ \\
$4::$ & $\sf{edge(a,c)}$ & $\ot$ & $0::$ & $\sf{reachable(a)}$ & $\ot \sf{reachable(c)}, \sf{edge(c,a)}$ \\
$20::$ & $\sf{edge(a,d)}$ & $\ot$ & $0::$ & $\sf{reachable(a)}$ & $\ot \sf{reachable(d)}, \sf{edge(d,a)}$\\
$9::$ & $\sf{edge(c,a)}$ & $\ot$ & $0::$ & $\sf{reachable(c)}$ & $\ot \sf{reachable(a)}, \sf{edge(a,c)}$  \\
$15::$ & $\sf{edge(c,d)}$ & $\ot$ & $0::$ & $\sf{reachable(c)}$ & $\ot \sf{reachable(d)}, \sf{edge(d,c)}$ \\
$16::$ & $\sf{edge(d,c)}$ & $\ot$ & $0::$ &
$\sf{reachable(d)}$ & $\ot \sf{reachable(a)}, \sf{edge(a,d)}$ \\
$0::$ & $\sf{reachable(a)}$ & $\ot \sf{initial(a)}$ & $0::$ & $
\sf{reachable(d)}$ & $\ot \sf{reachable(c)}, \sf{edge(c,d)}$ \\
$0::$ & $\sf{reachable(c)}$ & $\ot \sf{initial(c)}$  & & &
\end{tabular}
\end{center}
\begin{minipage}{0.5\textwidth}
The program describes the (single-sourced, directed) weighted graph (\ref{eq:simple-short-path-graph}) on the right:
\end{minipage}
\begin{minipage}{0.5\textwidth}
\begin{equation}\label{eq:simple-short-path-graph}
\vcenter{
\xymatrix{
a \ar@/_/[d]|{\whitespacearound{4}} \ar@/^/[dr]|{\whitespacearound{20}} &  \\
c \ar@/_/[r]_{\whitespacearound{15}} \ar@/_/[u]|{\whitespacearound{9}} & d \ar@/_/[l]|{\whitespacearound{16}}
}
}
\end{equation}
\end{minipage}
The (finite) set of atoms $\At_{gsp}$ is
\[
\{ \sf{initial(a)}, \dots, \sf{initial(c)}, \sf{edge(a,c)},\dots,  \sf{edge(d,c)}, \sf{reachable(a)}, \dots, \sf{reachable(c)}  \}.
\]
Some weights of the corresponding coalgebra $p^{gsp} \colon \At_{gsp}\to \multiFunc \fPowset(\At_{gsp})$ are presented below:
\begin{align*}
p^{gsp}(\sf{initial(a)}) & = 0 \{ \} \quad\quad p^{gsp}(\mathsf{edge(a,d)}) = 20 \{  \}  
\quad\quad p^{gsp}(\sf{edge(d,a)}) = \varnothing \\
p^{gsp}(\sf{reachable(a)}) & = 0 \{ \mathsf{initial(a)} \} + 0 \{ \mathsf{reachable(c)}, \mathsf{edge(c, a)} \} + 0 \{ \mathsf{reachable(d)}, \sf{edge(d, a)} \}
\end{align*}
Note that the $0$ in the above equations is the $\semiTimes$-unit $\prodUnit$ in $\semiK^{sp}$, and should not be confused with the $\semiPlus$-unit $\plusUnit$.
\end{exa}

Just as in the case of $\PLP$, we may represent operationally the recursive calculation of the weight $\weight(A)$ of a goal $A$ as a certain kind of tree, which we call \emph{weighted derivation tree}. These are weighted version of the and-or trees for pure logic programming \cite{gutpa1994}. They are formally defined analogously to the stochastic derivation trees in Definition \ref{def:stochastictrees_ground}, except that the probability labels are replaced by weight labels. We illustrate the concept with an example.

\begin{exa}\label{ex:ground-weight-der-tree}
In the context of Example \ref{ex:ground-WLP}, the weighted derivation tree for $\sf{reachable(a)}$ is partially depicted below. Note the different meaning for an atom-node and a clause-node to have no child. For instance, the clause-node following $\mathsf{init(a)}$ has no child because `$0 :: \mathsf{initial(a)} \ot$' is a clause in $\progGroundSP$, while the atom-node $\mathsf{edge(d,a)}$ has no child because there is no clause in $\progGroundSP$ whose head is $\mathsf{edge(d,a)}$.
\begin{equation}\label{eq:ex-WLP-der-tree}
\vcenter{
\xymatrix@R=5pt@C=3pt{
&&&& \mathsf{reach(a)} \ar@{-}[lllld]_{0} \ar@{-}[lld]^{0} \ar@{-}[rd]_{0} \ar@{-}[rrrrd]^{0} &&&& \\
\bullet \ar@{-}[d] && \bullet \ar@{-}[dl] \ar@{-}[dr] &&& \bullet \ar@{-}[dl] \ar@{-}[dr] &&& \bullet \ar@{-}[dl] \ar@{-}[dr] && \\
\mathsf{init(a)} \ar@{-}[d]_{0} & \mathsf{reach(a)} \ar@{.}[d] && \mathsf{edge(a,a)} & \mathsf{reach(c)} \ar@{.}[d] && \mathsf{edge(c,a)} \ar@{-}[d]_{9} & \mathsf{reach(d)} \ar@{.}[d] && \mathsf{edge(d,a)} \\
\bullet &&&& && \bullet &&&
}
}
\end{equation}

\end{exa}

By replacing $\F$ by $\multiFunc$ in Construction \ref{constr:terminalseq_ground}, one may construct the terminal sequence for the functor $\At\times \multiFunc \fPowset(-)$, which by accessibility \cite{worrell1999terminal} converges to a limit $X_\gamma$, yielding the final $\At\times \multiFunc \fPowset$-coalgebra $X_{\gamma} \cong \At\times \multiFunc \fPowset(X_{\gamma})$ and a unique coalgebra morphism $\wFinalSem{-}{w}$ as below:
\begin{equation*}
\xymatrix@R=15pt@C+75pt{
\At \ar@{-->}[r]^{\wFinalSem{-}{w}} \ar[d]^{<\id, w>} & X_{\gamma} \ar[d]^{\iso} \\
\At \times \multiFunc \fPowset (\At) \ar[r]^{\id \times \multiFunc \fPowset (\wFinalSem{-}{w})} & \At \times \multiFunc \fPowset(X_{\gamma}).
}
\end{equation*}
By construction, one may readily verify that weighted derivation trees are elements of $X_{\gamma}$, and $\wFinalSem{-}{w}$ maps an atom $A$ to its weighted derivation tree, computed according to the program $\progW$. Thus we can define the derivation semantics for $\WLP$ as the map $\wFinalSem{-}{w}$. Note the weight $\weight^{\scriptscriptstyle w}(A)$ of a goal is effectively computable from its semantics $\wFinalSem{A}{w}$ --- see Appendix \ref{app:compute-weight} for details.


\subsection{Coalgebraic Semantics for Arbitrary $\WLP$ Programs}\label{sec:WLP-generalcase}
In this subsection we generalise our approach to arbitrary weighted logic programs, possibly including variables and functions. We shall assume that the semiring $\semiK$ underlying the $\WLP$ programs is countably $\semiPlus$-complete, namely $\left( \semiSum_{c\in C} c \right) \in K$ and $d \semiTimes \left( \semiSum_{c\in C} c \right) = \semiSum_{c\in C} (d \semiTimes c)$, for arbitrary countable subset $C \subseteq K$. As we will see, the reason for such assumption is that there may exist countably many substitutions that match the head of some clause to the goal, generating countably many different instances of the body. The weight semantics requires adding up all these countably many weights, whence their sum should be an element of $\semiK$. 

The construction leading to a coalgebraically defined derivation semantics $\wFinalSemGen{-}{}$ for $\WLP$ is fairly similar to the one for $\PLP$ (Subsection \ref{sec:coalg_general}), thus we confine ourselves to highlighting the main steps and emphasise the differences between the two approaches. 

The derivation semantics $\wFinalSemGen{-}{}$ will associate a goal with its resolution by unification in a program. This resolution will be represented as a tree, which we call \emph{weighted saturated derivation tree}. Formally, these are defined analogously to stochastic saturated derivation trees (Definition~\ref{def:stoch_saturate_tree}), with the difference that there is no need of clause-nodes (for reasons detailed in Remark~\ref{rmk:diffprobandweight} below), and labels on edges departing from substitution-nodes are now weights rather than probabilities. 

Before providing a coalgebraic modelling of the semantics, we establish some preliminary intuition via our leading example.
\begin{exa}\label{ex:gen-weight-der-tree}
Recall program $\progShortPath$ from Example \ref{ex:weight-short-path}. We show below part of the weighted saturated derivation tree for the goal $\mathsf{reachable(x)}$ (in context $1$), where we write $\varnothing$ for the empty substitution. Compared with the ground case (see \eqref{eq:ex-WLP-der-tree} in Example~\ref{ex:ground-weight-der-tree}), the tree in \eqref{eq:ex-sat-weighted-der-tree} has an extra layer of `substitution-nodes', labelled with the substitutions applied to the goal during resolution by unification. For instance, the node labelled with $\sf{x \mapsto b}$ has one child for each substitution instance of the body $\{ B_1,\dots, B_k \}\tau$ satisfying that $c:: H\ot B_1,\dots, B_k$ is a clause in $\progShortPath$ and  the head $H$ term-matches the goal $\sf{reach(x)}[\sf{x \mapsto b}] = \sf{reach(b)}$ via substitution $\tau$. 
\begin{equation}\label{eq:ex-sat-weighted-der-tree}
\vcenter{
\xymatrix@R=5pt@C=5pt{
&&&& \mathsf{reach(x)} \ar@{-}[dllll] \ar@{-}[d] \ar@{-}[drrrr] \ar@{-}[drrrrrr] &&&&&& \\
\mathsf{x\mapsto a} \ar@{.}[d] &&&& \mathsf{x\mapsto b} \ar@{-}[dlll]_{0} \ar@{-}[dl]^{0} \ar@{-}[dr] \ar@{-}[drrr]^{0} &&&& \mathsf{x\mapsto x} \ar@{.}[d] && \cdots & \\
& \instanceNode \ar@{-}[d] && \instanceNode \ar@{-}[dl] \ar@{-}[dr] && \cdots && \instanceNode \ar@{-}[dl] \ar@{-}[dr] &&& \\
& \mathsf{init(b)} \ar@{-}[d] & \mathsf{reach(a)} \ar@{.}[d] && \mathsf{edge(a,b)} \ar@{.}[d] && \mathsf{reach(d)} \ar@{.}[d] && \mathsf{edge(d,b)} \ar@{.}[d] && \\
& \varnothing &&&&&&&&& \\
}
}
\end{equation}
 As mentioned, with respect to stochastic saturated derivation trees (see e.g.~\eqref{eq:exampletreegeneral}), the semantics represented in \eqref{eq:ex-sat-weighted-der-tree} does not require the presence of clause-nodes (which were labelled with $\clauseCnode$ in~\eqref{eq:exampletreegeneral}).
\end{exa}
Similarly to the case of $\PLP$, we will describe resolution by unification in two steps: first, one considers all substitutions $\theta$ of the goal $A$. Then, each substitution instance $A \theta$ of the goal is compared by term-matching to the heads of clauses in the program $\progW$. The function $u$ performing term-matching of $A\theta$ in $\progW$ associates a set $\{C_1,\dots,C_m\}$ of atoms with an element of $\semiK$, computed as follows:
\begin{equation}\label{eq:WLP-term-match}
u(A \theta) \colon \{ C_1,\dots, C_m \} \mapsto \semiSum_{
\substack{
c :: H \ot B_1, \dots, B_k \text{is a clause of $\progW$}\\ \text{such that } \exists \tau \text{ with } ~ H\tau = A\theta \\
\text{and } \{B_1,\dots, B_k\}\tau = \{ C_1,\dots, C_m \}
}} c
\end{equation}
Intuitively, for each clause $c :: H \ot B_1,\dots, B_k$ in $\progW$ and substitution $\tau$ that matches the head $H$ with $A\theta$, $u(A\theta)$ assigns weight $c$ to $\{ B_1, \dots, B_k \} \tau$. There are three observations to make here. First, it is possible that there is a different clause $c' :: H' \ot B'_1,\dots, B'_{k'}$ and substitution $\tau'$ such that $H'\tau' = A\theta$, and $\{ B'_1,\dots, B'_{k'} \}\tau' = \{ B_1,\dots, B_k \}\tau = \{ C_1,\dots, C_m \}$. In this case, the weight $c'$ is also added to the value of $u(A\theta)(\{ C_1,\dots, C_m \})$. Second, there may exist a different substitution $\sigma$ such that $H\tau = H\sigma$ and $\{ B_1,\dots, B_k \}\tau = \{ B_1,\dots, B_k \}\sigma$. In this case, no extra weight $c$ is added to $u(A\theta)(\{ B_1,\dots, B_k \}\tau)$. Third, the support $\supp(u(A\theta))$ of the function might be countably infinite. This is due to the possible existence of a clause $H\ot B_1,\dots, B_k$ and countably many substitutions $\tau_1, \tau_2,\dots$ such that $H\tau_i = A\theta$ while $\{ B_1,\dots, B_k \}\tau_i \neq \{ B_1,\dots, B_k \}\tau_j$, for all distinct $i, j = 1,2,\dots$. 

For example, the application of $u$ to $\sf{reach(b)}$ in Example~\ref{ex:gen-weight-der-tree} can be visualised in tree form as the first two layers of the subtree starting at $\sf{x \mapsto b}$. 

Towards a categorical treatment of resolution by unification, we need to introduce the \emph{countable} (because of the third observation above) multiset functor $\cMultiFunc \colon \sets\to \sets $ over a $\semiPlus$-complete semiring $\semiK$. On objects, $\cMultiFunc (X) = \{  \phi \colon X\to \semiK \mid \supp(\phi) ~{\text{is at most countable}} \}$. On morphisms, $\cMultiFunc(h) \colon \cMultiFunc (X) \to \cMultiFunc (Y)$ maps $\phi$ to $\lambda (y\in Y). \semiSum_{x\in h^{-1}(y)} \phi(x)$. Now, lifting functors from $\sets$ to $\sets^{|\LawOp|}$ as in \eqref{eq:TMcoalg}, the term-matching part $u$ of the resolution can be expressed as a $\lift{\cMultiFunc} \lift{\fPowset}$-coalgebra in $\sets^{|\LawOp|}$:
\begin{equation}\label{eq:WLP-TMcoalg}
	\coalgU \colon \U\At\to \lift{\cMultiFunc} \lift{\fPowset} \U\At
\end{equation}

\begin{rem}\label{rmk:diffprobandweight} Note the coalgebra type of term-matching for $\WLP$ ($\lift{\cMultiFunc} \lift{\fPowset}$) is simpler than the one for $\PLP$ ($\liftF \lift{\powset_c} \lift{\powset_f}$), as in \eqref{eq:TMcoalg}. In fact, in the distribution semantics of $\PLP$, the use of a clause in a proof should be counted at most once. This is taken care by the extra layer $\liftF$ which allows to identify different substitution instances of the same clause in a single tree. Instead, in the weight semantics every use of a clause counts towards calculating the total weight (\emph{cf.} \eqref{eq:cal-weight}). Thus it is unnecessary to keep track of whether clauses are used more than once. 

 This difference is reflected in the tree representation of the two semantics, compare e.g. \eqref{eq:ex-sat-weighted-der-tree} with \eqref{eq:exampletreegeneral}. In the weighted saturated trees, there is no need for clause-nodes, which are instead present in stochastic saturated trees. \end{rem}

Analogously to the $\PLP$ case, once term matching is modelled as a coalgebra, we can then exploit the unit $\eta$ of the adjunction $\funcU \dashv \funcK$ (see \eqref{eq:saturationadju}) to model the retrieval of all the substitutions of the goals. This will yield a $\funcK \lift{\cMultiFunc} \lift{\fPowset} \funcU$-coalgebra in $\sets^{\LawOp}$, capturing resolution by unification:
\begin{equation}\label{eq:w-sharp-compose}
\saturate{\coalgU} \coloneqq \At \xto{\eta_{\At}} \funcK \funcU \At \xto{\funcK \coalgU} \funcK \lift{\cMultiFunc} \lift{\fPowset} \funcU \At
\end{equation}
Spelling out the definitions, given some $A \in \At(n)$,
\[
\saturate{\coalgU}_n \colon A \mapsto \lr{ \coalgU_m (A\theta) }_{\theta \in \LawOp[n,m]}
\]

Finally, we can formulate the derivation semantics for general $\WLP$ via the cofree coalgebra. First, one may generate the terminal sequence for the functor $\At\times \funcK \lift{\cMultiFunc} \lift{\fPowset} \funcU$, following a similar procedure as Construction \ref{constr:terminalseq_general}. Since $\cMultiFunc$ is accessible, the terminal sequence converges, say at $Z_\gamma$. Then the derivation semantics $\intp{-}_{\saturate{\coalgU}}$ is defined by the universal mapping property of the final coalgebra:
\begin{equation}
  \vcenter{
\xymatrix@R=15pt@C+75pt{
\At \ar@{-->}[r]^{\intp{-}_{\saturate{\coalgU}}} \ar[d]_{ <id , \saturate{\coalgU} >}  & Z_\gamma \ar[d]^{\iso } \\
\At\times \funcK \lift{\cMultiFunc} \lift{\fPowset} \funcU \At \ar[r]_{id \times  \funcK \lift{\cMultiFunc} \lift{\fPowset} \funcU \intp{-}_{\saturate{\coalgU}} } & \At\times \funcK \lift{\cMultiFunc} \lift{\fPowset} \funcU Z_\gamma
}}
\end{equation}
As expected, weighted saturated derivation trees can be regarded as elements of the final coalgebra $Z_\gamma$, and $\intp{-}_{\saturate{\coalgU}}$ maps an atom $A$ to its weighted saturated derivation tree. 

\section{Conclusion}\label{sec:conclusion}
This work proposed a coalgebraic semantics for probabilistic logic programming and weighted logic programming, as an extension of the framework in \cite{komendantskaya2010coalgebraic,BonchiZ15} for pure logic programming. Our approach consists of the following three steps. First we represent the programs as coalgebras of appropriate type. Next we define the derivation semantics in terms of certain derivation trees, which can be recovered as objects in the corresponding final coalgebra. Finally, we explain how to retrieve the semantics that is most commonly found in the literature (distribution semantics for $\PLP$ and weight semantics for $\WLP$) from our coalgebraic semantics. While the first two steps simply generalise  the approach in \cite{BonchiZ15}, the last step has no counterpart in pure logic programming, and it requires extra categorical constructions (distribution trees) and computations (Appendix \ref{sec:algo_ground}, \ref{app:general_computation}, \ref{app:compute-weight}).

In extending our coalgebraic semantics from ground to arbitrary $\PLP$ programs, we adopted the `saturated' approach of \cite{BonchiZ15}. The `lax naturality' approach of \cite{komendantskaya2011coalgebraic} offers an alternative route (see \cite{komen2018logic} for a detailed comparison), based on the observation that, even though first-order logic programs fail to be natural transformations, they are lax natural in a suitable order enriched category. Developing a coalgebraic lax semantics for $\PLP$ and $\WLP$ would yield  derivation trees representing computations by term-matching rather than general unification. The extra information given by lax naturality could be used for algorithmic purposes, similarly to what is done in \cite{KomendantskayaSH14,KomendantskayaPS13} for pure logic programming. We leave the development of lax semantics for $\PLP$ and $\WLP$ as future work.

As follow-up work, we plan to investigate in two main directions. First, starting from the observation that Bayesian networks (over binary variables) can be seen as $\PLP$ programs, we want to understand Bayesian reasoning within our coalgebraic framework. Recent works on a categorical semantics for Bayesian probability (in particular \cite{JacobsZ19,ChoJ19}) provide a starting point for this research. A major challenge is that standard translations of Bayesian networks into $\PLP$ (such as the one in \cite{Poole08}) require to incorporate negation in the logical syntax. Thus a preliminary step will be to model coalgebraically $\PLP$ with negation, which we plan to do following the credal set approach outlined in \cite{cozman2019}.

A second research thread is investigating the semantics of coinductive logic programs in the quantitative setting of $\PLP$ and $\WLP$. This is an essentially unexplored perspective, which we will tackle building on research on coinduction in the coalgebraic semantics of standard logic programming \cite{Komendantskaya18PositionPaper,
GuptaCoinductiveLogicProgramming,
basold2019}. 

\bibliographystyle{alpha}
\bibliography{plp-lmcs}

\newpage
\appendix

\section{Computability of the Distribution Semantics (Ground Case)} \label{sec:algo_ground}
\smallskip\noindent\textbf{Computing with distribution trees.} As a justification for our tree representation of the distribution semantics, we claimed that the probability $\Pr_\Prog(A)$ associated with a goal (see \eqref{eq:distrsemantics}) can be straightforwardly computed from the corresponding distribution tree $\intd{A}_p$. This appendix supplies such an algorithm. Note that this serves just as a proof of concept, without any claim of efficiency compared to pre-existing implementations. 
In this section we fix a ground $\PLP$ program $\Prog$ with atoms $\At$, a goal $A \in \At$ and the distribution tree $\mathcal{T}$ for $A$ in $\Prog$ (Definition \ref{def:distrtree_ground}).
First, we may assume that the distribution tree $\mathcal{T}$ is finite. Indeed, in the ground case infinite branches only result from loops, namely multiple appearances of an atom in some path, which can be easily detected. We can always prune the subtrees of $\mathcal{T}$ rooted by atoms that already appeared at an earlier stage: this does not affect the computation of $\Pr_\Prog(A)$, and it makes $\mathcal{T}$ finite.
Next, we introduce the concept of \emph{deterministic} subtree. Basically a deterministic subtree selects one world-node at each stage. Recall that every clause-node in $\mathcal{T}$ represents a clause in $\pureProg{\Prog}$, whose head is the label of its atom-grandparent, and body consists of the labels of its atom-children.
\begin{defi}\label{def:det_subtree}
A subtree $\mathcal{S}$ of $\mathcal{T}$ is \emph{deterministic} if (i) it contains exactly one child (world-node) for each atom-node and all children for other nodes, and (ii) for any distinct atom-nodes $s,t$ in $\mathcal{S}$ with the same label, $s$ and $t$ have their clause-grandchildren representing the same clauses.
\end{defi}
The idea is that $\mathcal{S}$ describes a computation in which the choice of a possible world (i.e., a sub-program of $\Prog$) associated to any atom $B$ appearing during the resolution is uniquely determined. Because of this feature, each deterministic subtree uniquely identifies a set of sub-programs of $\Prog$, and together the deterministic subtrees of $\mathcal{T}$ form a \emph{partition} over the set of these sub-programs (see Proposition~\ref{prop:det_subtree} below). 

Since $\mathcal{T}$ is finite, it is clear that we can always provide an enumeration of its deterministic subtrees. We can now present our algorithm, in two steps. First, Algorithm \ref{alg.prob.subtree} computes the probability associated with a deterministic subtree. Second, Algorithm \ref{alg.prob.goal} computes $\Pr_{\Prog}(A)$ by summing up the probabilities found by Algorithm \ref{alg.prob.subtree} on all the deterministic subtrees of $\mathcal{T}$ that contains a refutation of $A$. Below we write `$\mathrm{label}(s \to t)$' for the probability value labelling the edge from $s$ to $t$.
\begin{algorithm}[H]
\caption{Compute probability of a deterministic subtree}
\label{alg.prob.subtree}
\hspace*{\algorithmicindent} \textbf{Input:} A deterministic subtree $\mathcal{S}$ of $\mathcal{T}$\\
\hspace*{\algorithmicindent} \textbf{Output:} The probability of $\mathcal{S}$\\
\begin{algorithmic}[1]
	\State{prob\_list = [\ ]}
	\For{atom-node $s$ in $\mathcal{S}$}
		\If{$s$ has child}
			\State{prob\_list.append(label($s \to \mathrm{child}(s)$)) }
		\EndIf
	\EndFor
	\If{prob\_list == [\ ]}
		\State{\Return 0}
	\Else
		~prob = product of values in prob\_list
		\State{\Return prob}
	\EndIf
\end{algorithmic}
\end{algorithm}
\begin{algorithm}[H]
\caption{Compute probability of a goal}
\label{alg.prob.goal}
\hspace*{\algorithmicindent} \textbf{Input:} The distribution tree $\mathcal{T}$ of $A$ in $\Prog$\\
\hspace*{\algorithmicindent} \textbf{Output:} The success probability $\Pr_\Prog(A)$\\
\begin{algorithmic}[1]
	\State{prob\_suc = 0}
	\For{deterministic subtree $\mathcal{S}$ of $\mathcal{T}$}
		\If{$\mathcal{S}$ proves $A$}
			\State{prob\_suc += \textbf{Algorithm \ref{alg.prob.subtree}}($\mathcal{S}$)}
		\EndIf
	\EndFor
	\State{\Return prob\_suc}
\end{algorithmic}
\end{algorithm}
The above procedure terminates because $\mathcal{T}$ is finite and every for-loop is finite. We now focus on the correctness of the algorithm.

\smallskip\noindent\textbf{Correctness.} 
As mentioned, a world-node in a deterministic subtree can be seen as a choice of clauses: one chooses the clauses represented by its clause-children, and discards the clauses represented by its `complement' world. For correctness, we make  this precise, via the following definition.
\begin{defi}\label{def:accept_reject}
Given a clause $\clauseC$ in $\Prog$, a deterministic subtree $\mathcal{S}$ of $\mathcal{T}$, a world-node $t$ and its atom-parent $s$ in $\mathcal{S}$, we say $t$ \emph{accepts} $\clauseC$ if $\head{\clauseC} = \mathrm{label}(s)$ and there is a clause-child of $t$ that represents $\clauseC$; $t$ \emph{rejects} $\clauseC$ if $\head{\clauseC} = \mathrm{label}(s)$ but no clause-child of $t$ represents $\clauseC$. We say $\mathcal{S}$ \emph{accepts} (\emph{rejects}) $\clauseC$ if there exists a world-node $t$ in $\mathcal{S}$ that accepts (rejects) $\clauseC$.
\end{defi}
Note that Definition \ref{def:det_subtree}, condition (ii) prevents the existence of world-nodes $t,t'$ in $\mathcal{S}$ such that $t$ accepts $\clauseC$ and $t'$ rejects $\clauseC$. Thus the notion that $\mathcal{S}$ accepts (rejects) $\clauseC$ is well-defined. We denote the set of clauses accepted and rejected by $\mathcal{S}$ by $\Acc(\mathcal{S})$ and $\Rej(\mathcal{S})$, respectively. Then we can define the set $\repSubProg{\mathcal{S}}$ of sub-programs  represented by $\mathcal{S}$ as
\begin{equation} \label{eq:def_repSubProg}
\repSubProg{\mathcal{S}} \coloneqq \{ \progL \subseteq |\Prog| \mid \forall \clauseC \in \Acc(\mathcal{S}), \clauseC \in \progL;  \forall \clauseC' \in \Rej(\mathcal{S}), \clauseC' \notin \progL \}
\end{equation}
We will prove the correctness of the algorithm through the following basic observations on the connection between deterministic subtrees and the sub-programs they represent:
\begin{prop}\label{prop:det_subtree}
Suppose $\mathcal{S}$ is a deterministic subtree of the distribution tree $\mathcal{T}$ of $A$.
\begin{enumerate}
	\item $\{ \repSubProg{\mathcal{S}} \mid \mathcal{S} \text{ is deterministic subtree of } \mathcal{T}\}$ forms a partition of $\powset(\Prog)$.
	\item Either $\progL \vdash A$ for all $\progL \in \repSubProg{\mathcal{S}}$ or $\progL \not\vdash A$ for all $\progL \in \repSubProg{\mathcal{S}}$.
	\item $\sum_{\progL \in \repSubProg{\mathcal{S}}} \Pr_\Prog(\progL) = \prod_{r_i \in \mathcal{S}} r_i$, where the $r_i$s are all the probability labels appearing in $\mathcal{S}$ (on the $\text{atom-node}\to \text{world-node}$ edges).
\end{enumerate}
\end{prop}

\begin{proof}~
\begin{enumerate}
	\item Given any two distinct deterministic subtrees, there is an atom-node $s$ such that the subtrees include distinct world-child of $s$. So by \eqref{eq:def_repSubProg} the sub-programs they represent do not share at least one clause. Moreover, given a sub-program $\progL$, one can always identify a deterministic subtree $\mathcal{S}$ such that $\progL \in \repSubProg{\mathcal{S}}$, as follows: given the $A$-labelled root of $\mathcal{T}$, select the world-child $w$ of $A$ representing the (possibly empty) set $X$ of all clauses in $\progL$ whose head is $A$; then select the children (if any) of $w$, and repeat the procedure. 
	\item Note that a sub-program $\progL \in \repSubProg{\mathcal{S}}$ proves the goal $A$ iff $\mathcal{S}$ contains a successful refutation of $A$, and the latter property is independent of the choice of $\progL$. 
	\item  We refer to $\prod_{r_i \in \mathcal{S}} r_i$ as the probability of the deterministic subtree $\mathcal{S}$. For each sub-program $\progL \in \repSubProg{\mathcal{S}}$, its probability can be written as
	\begin{equation}\label{eq:prob_subprogram}
	\progPr_\Prog(\progL) = \left( \prod_{\clauseC\in \Acc(\mathcal{S})} \lab{\clauseC} \right) \cdot \left( \prod_{\clauseC'\in \Rej(\mathcal{S})} \left( 1 - \lab{\clauseC'} \right) \right) \cdot \progPr_{\Prog\setminus (\Acc \cup \Rej)} (\progL \setminus \Acc(\mathcal{S}))
	\end{equation}
	Note that $\repSubProg{\mathcal{S}}$ can also be written as $\{ X \cup \Acc(\mathcal{S}) \mid X\subseteq \Prog \setminus (\Acc(\mathcal{S}) \cup \Rej(\mathcal{S})) \}$, so
	\begin{equation}\label{eq:sum_is_one}
	\sum_{\progL\in \repSubProg{\mathcal{S}}} \progPr_{\Prog\setminus (\Acc \cup \Rej)} (\progL\setminus \Acc(\mathcal{S})) = 1.
	\end{equation}
	Applying equation \eqref{eq:sum_is_one} to the sum of \eqref{eq:prob_subprogram} over all $\progL\in \repSubProg{\mathcal{S}}$, we get
	\begin{equation}\label{eq:sum_subprog}
	\sum_{\progL \in \repSubProg{\mathcal{S}}} \progPr_{\Prog} (\progL)  = \prod_{\clauseC\in \Acc(\mathcal{S})} \lab{\clauseC} \cdot \prod_{\clauseC'\in \Rej(\mathcal{S})} (1 - \lab{\clauseC'})	
	\end{equation}
	For each world-node $t$ and its atom-parent $s$, we can use the terminology in Definition \ref{def:accept_reject}, and express $\mathrm{label}(s\to t)$ (see Definition \ref{def:distrtree_ground}) as
	\begin{equation}\label{eq:label_as_accept_reject}
	\mathrm{label}(s\to t) = \prod_{t\text{ accepts }\clauseC}\lab{\clauseC} \cdot \prod_{t\text{ rejects }\clauseC'}(1 - \lab{\clauseC'}).
	\end{equation}
	Applying \eqref{eq:label_as_accept_reject} to the whole deterministic subtree $\mathcal{S}$, we obtain  
	\begin{align*}
	\sum_{\progL\in \repSubProg{\mathcal{S}}} \progPr_\Prog(\progL)  \quad & \eql{\eqref{eq:sum_subprog}} \prod_{\clauseC\in \Acc(\mathcal{S})} \lab{\clauseC} \cdot \prod_{\clauseC'\in \Rej(\mathcal{S})} (1 - \lab{\clauseC'}) \\
	& \eql{\mathrm{Def. } \ref{def:accept_reject}} \prod_{(\text{world-node $t$ in } \mathcal{S})} \left( \prod_{t\text{ accepts }\clauseC}\lab{\clauseC} \cdot \prod_{t\text{ rejects }\clauseC'}(1 - \lab{\clauseC'}) \right) \\
	& \eql{\eqref{eq:label_as_accept_reject}} \prod_{r_i\in \mathcal{S}} r_i
	\end{align*}
\end{enumerate}
If we say two world-nodes $t$ and $t'$ are equivalent if their clause-children represent exactly the same clauses in $\Prog$, then the $\prod_{(\text{world-node $t$ in } \mathcal{S})}$ in the above calculation visits every world-node exactly once modulo equivalence.
\end{proof}
We can now formulate the success probability of $A$ as follows
\begin{align*}
\Pr_\Prog (A) & = \sum_{\pureProg{\Prog} \supseteq \progL \vdash A} \progPr_\Prog(\progL)  \quad\quad \eql{(\mathrm{Prop.}\ref{prop:det_subtree}, 1\& 2)} \quad\quad \sum_{\mathcal{S} \vdash A} \sum_{\progL\in \repSubProg{\mathcal{S}}} \progPr_\Prog(\progL) \\
& \eql{\eqref{eq:sum_subprog}} \quad \sum_{\mathcal{S} \vdash A} \left( \prod_{\clauseC\in \Acc(\mathcal{S})} \lab{\clauseC} \cdot \prod_{\clauseC'\in \Rej(\mathcal{S})} (1 - \lab{\clauseC'}) \right) \qquad  \eql{(\mathrm{Prop. }\ref{prop:det_subtree}, 3)}  \qquad    \sum_{\mathcal{S} \vdash A} \prod_{r_i\in S} r_i 
\end{align*}
In words, this is exactly Algorithm \ref{alg.prob.goal}: we sum up the probabilities of all deterministic subtrees $\mathcal{S}$ of the distribution tree $\mathcal{T}$ which contain a proof of $A$.

\section{Computability of the Distribution Semantics (General Case)} \label{app:general_computation}
Computability of the distribution semantics for arbitrary $\PLP$ programs relies on the substitution mechanism employed in the resolution. This aspect deserves a preliminary discussion. Traditionally, logic programming has both the theorem-proving and problem-solving perspectives \cite{komendantskaya2018}. From the problem-solving perspective, the aim is to find a refutation of the goal $\ot G$, which amounts to finding a proof of \emph{some substitution instance} of $G$. From the theorem-proving perspective, the aim is to search for a proof of the goal $G$ itself as an atom. The main difference is in the substitution mechanism of resolution: unification for the problem-solving and term-matching for the theorem-proving perspective. 
We will first explore computability within the theorem-proving perspective. As resolution therein is by term-matching, the probability $\probTM(A)$ of proving a goal $A$ in a $\PLP$ program $\Prog$ is formulated as 
$\probTM(A) \coloneqq \sum\limits_{\pureProg{\Prog} \supseteq \progL \proves A} \progPr_\Prog(\progL)$, 
where $\progL \proves A$ means that $A$ is derivable in the sub-program $\progL$ (not to be confused with $\progL\vdash A$, which stands for \emph{some substitution instance} of $A$ being derivable in $\progL$, see~\eqref{eq:distrsemantics}).


In our coalgebraic framework, the distribution semantics for general $\PLP$ programs is represented on `saturated' trees, in which computations are performed by unification. However, following \cite{BonchiZ15}, one can define the \emph{TM (\textbf{T}erm \textbf{M}atching) distribution tree} of a goal $A$ in a program $\Prog$ by `desaturation' of the saturated distribution tree for $A$ in $\Prog$. The coalgebraic definition, for which we refer to \cite{BonchiZ15}, applies pointwise on the saturated tree the counit $\epsilon_{\U\At} \colon \U\K\U\At \to \U\At$ of the adjunction $\U\dashv \K$ (\emph{cf.}~\eqref{eq:saturationadju}). The TM distribution tree which results from `desaturation' can be described very simply: at each layer of the starting saturated distribution tree, one prunes all the subtrees which are not labelled with the identity substitution $\id := x_1 \mapsto x_1, x_2 \mapsto x_2, \dots$. In this way, the only remaining computation are those in which resolution only applies a non-trivial substitution on the clause side, that is, in which unification is restricted to term-matching.

\smallskip\noindent\textbf{Computability of term-matching distribution semantics.} One may compute the success probability $\probTM(A)$ in $\Prog$ from the TM distribution tree of $A$ in $\Prog$. The computation goes similarly to Algorithm~\ref{alg.prob.goal} : the problem amounts to calculating the probabilities of those deterministic subtrees of the distribution tree which prove the goal. We confine ourselves to some remarks on the aspects that require extra care, compared to the ground case.
\begin{enumerate}
\item The probability $\probTM(A)$ is not computable in whole generality. It depends on whether one can decide all the proofs of $A$ in the pure logic program $\pureProg{\Prog}$, and there are various heuristics in logic programming for this task.
	\item It is still possible to decide whether a subtree is deterministic, but the algorithm in the general case is a bit subtler, as it is now possible that two different goals match the same clause (instantiated in two different ways).
	\item When calculating the probability of a deterministic subtree in the TM distribution tree, multiple appearances of a single clause (possibly instantiated with different substitutions) should be counted only once. In order to ensure this one needs to be able to identify which clause is applied at each step of the computation described by the  distribution trees: this is precisely the reason of the addition of the clause labels in the coalgebra type of these trees, as discussed in Section \ref{sec:gen_dist_semantics}.
\end{enumerate}

We conclude by briefly discussing the problem-solving perspective, in which resolution is based on arbitrary unification rather than just term-matching. In standard SLD-resolution, computability relies on the possibility of identifying the \emph{most general} unifier between a goal and the head of a given clause. This can be done also within saturated distribution trees, since saturation supplies \emph{all} the unifiers, thus in particular the most general one. This means that, in principle, one may compute the distribution semantics based on most general unification from the saturated distributed tree associated with a goal, with similar caveats as the ones we described for the term-matching case. However, the lack of a satisfactory coalgebraic treatment of most general unifiers \cite{BonchiZ15} makes us prefer the theorem-proving perspective discussed above, for which desaturation provides an elegant categorical formalisation.

\section{Computability of the weight semantics}\label{app:compute-weight}
In this appendix we will show how the weight $\weight(A)$ of a goal $A$ is computable from its derivation semantics, i.e. the weighted derivation tree $\wFinalSem{A}{w}$. We focus on the ground case, and then conclude with a discussion on how the generalisation to the variable case goes.

One preliminary consideration concerns cycles in $\WLP$. Indeed, in a weighted derivation tree it may be the case that an atom-node has a directed path to another atom-node labelled with the same atom. While different ways to assign meaning to cycles may be justifiable (for instance in the context of (co)inductive logic programming), the standard choice (see e.g. \cite{shay2010}) is to regard any proof subtree containing a cycle as failing to prove the given goal. In concrete, this means that proof subtrees with cycles may be simply discarded when calculating the weight of the goal. 

To make this precise in our computation, recall that a proof subtree for a goal $A$ is a subtree of $\wFinalSem{A}{w}$ where we select exactly one child for each atom-node with children. We may define the weight of a proof subtree for $A$ as the weight of $A$ in that particular tree, calculated using \eqref{eq:cal-weight}. Note the overall weight of a goal $A$ in $\wFinalSem{A}{w}$ is then the $\semiPlus$-sum of the weights of its proof subtrees. 

This means that, from a computational viewpoint, discarding a proof subtree is equivalent to assigning weight $\plusUnit$ (the $\semiPlus$-unit) to that proof subtree. In order to assign weight $\plusUnit$ to each proof subtree with cycles, we look for the first appearance of a cycle, re-label the next edge with $\plusUnit$, and discard all the descendants. This works because, according to \eqref{eq:cal-weight} and the fact that $\plusUnit \semiTimes c = c \semiTimes \plusUnit = \plusUnit$ for any $c \in K$, any finite proof subtree containing an edge labelled with $\plusUnit$ has weight $\plusUnit$. We illustrate this idea with the following example.

\begin{exa}\label{ex:cycle-ground}
Recall $\progGroundSP$ in Example \ref{ex:ground-WLP}, and consider the goal $\sf{reachable(c)}$. Its weighted derivation tree $\wFinalSem{\sf{reachable(c)}}{\progGroundSP}$ includes a finite path (subtree) witnessing a proof of the goal.
\begin{equation*}
\mathsf{reach(c)} \xto{4} \bullet \to \mathsf{reach(a)}  \xto{0} \bullet \to \mathsf{init(a)} \xto{0} \bullet
\end{equation*}
Intuitively, this says that $\sf{c}$ is reachable from the initial state $\sf{a}$ in in \eqref{eq:simple-short-path-graph}, with weight $4$. However, in $\wFinalSem{\sf{reachable(c)}}{\progGroundSP}$ there are also other proof subtrees, which feature cycles, as for instance 
\begin{gather}\label{eq:short-path-ground-revisit}
\mathsf{reach(c)} \xto{4} \bullet \to \mathsf{reach(a)} \xto{9} \bullet \to \mathsf{reach(c)} \xto{4} \bullet \to \mathsf{reach(a)} \xto{0} \bullet \to \mathsf{init(a)} \xto{0} \bullet \\ \label{eq:short-path-ground-revisit2}
\mathsf{reach(c)} \xto{4} \bullet \to \mathsf{reach(a)} \xto{9} \bullet \to \mathsf{reach(c)} \xto{4} \bullet \to \mathsf{reach(a)} \xto{9} \bullet \to \mathsf{reach(c)} \to \cdots
\end{gather}
Our algorithm for computing the weight of $w(\sf{reachable(c)})$ will prune both \eqref{eq:short-path-ground-revisit} and \eqref{eq:short-path-ground-revisit2}, letting them both become equal to
\begin{equation}\label{eq:short-path-ground-revisit-cut}
\mathsf{reach(c)} \xto{4} \bullet \to \mathsf{reach(a)} \xto{+\infty} \bullet 
\end{equation}
where $+\infty$ is the element $\plusUnit$ in the semiring of $\progGroundSP$. To see that this procedure yields the correct result, note that the weight of the proof subtree \eqref{eq:short-path-ground-revisit-cut} is $4 + (+\infty) = +\infty$, which does not affect the calculation of the weight ($= 4$) of $\sf{reachable(c)}$.
\end{exa}

\medskip
Now we give the generic algorithms which compute the weight of a goal $A$ given its weighted derivation tree $\wFinalSem{A}{w}$ in the program $\progW$. For readability, we divide our approach into 2 algorithms.
\begin{algorithm}[H]
\caption{Prune a weighted derivation tree}
\label{alg:prune-ground-weighted-tree}
\hspace*{\algorithmicindent} \textbf{Input:} a weighted derivation tree $\treeT$\\
\hspace*{\algorithmicindent} \textbf{Output:} $\treeT$ pruned all cycles\\
\begin{algorithmic}[1]
	\For{path $\pi$ in $\treeT$}
		\State{atom\_list = [ ]}
		\For{node $v$ in $\pi$}
		 	\If{$\mathsf{label}(v)$ in atom\_list}
		 		\State{$t = \text{grandparent of}~ v$, $u = \text{parent of}~ v$}
			 	\State{cut the descendants of $u$ from $\mathcal{T}$}
			 	\State{re-label the edge $t\to u$ by $\plusUnit$}
		 	\Else{}
		 		\State{atom\_list.append($\mathsf{label}(v)$)}
		 \EndIf
		\EndFor
	\EndFor
	\State{\Return tree $\treeT$}
\end{algorithmic}
\end{algorithm}

\begin{algorithm}[H]
\caption{Calculate weight of the root node from a pruned weighted derivation tree}
\label{alg:calulate-pruned-weighted-tree}
\hspace*{\algorithmicindent} \textbf{Input:} the weighted derivation tree $\treeT$ for the goal $A$\\
\hspace*{\algorithmicindent} \textbf{Output:} the weight of $A$\\
\begin{algorithmic}[1]
	\State{Prune $\treeT$ using Algorithm \ref{alg:prune-ground-weighted-tree}}
	\Function{Weight}{$x, \treeT$}
		\State{sum\_list = [ ]}
		\For{children $y$ of $x$}
			\State{prod\_list = [$\mathbf{label}(x\to y)$] \Comment{first get the weight of the clause} }
			\For{children $z$ of $y$}
				\State{prod\_list.append(W{\scriptsize{EIGHT}}($z, \treeT$))}
				\State{cur\_weight = $\semiProd$ prod\_list  \Comment{calculate the $\semiTimes$-product of all the values in prod\_list} }
			\EndFor
			\State{sum\_list.append(cur\_weight)}
		\EndFor
		\State{\textbf{return} $\semiSum$ sum\_list  \Comment{return the $\semiPlus$-sum of all the values in product\_list} }
	\EndFunction
	\State{$s$ = root of $\treeT$ \Comment{calculate the weight of node $s$ in $\treeT$}} 
	\State{\Return W{\scriptsize{EIGHT}}($s, \treeT$)}
\end{algorithmic}
\end{algorithm}
Algorithm \ref{alg:prune-ground-weighted-tree} uses depth-first search to find the first node whose atom label already appeared in some of its ancestors. This search procedure terminates because the set $\At$ of atoms is finite. Since a weighted derivation tree is finitely branching, the whole algorithm also terminates.

Algorithm \ref{alg:calulate-pruned-weighted-tree} is a typical divide-and-conquer algorithm, where the weight of an atom-node is calculated via the weights of its (atom-)grandchildren. The procedure is simply spelling out \eqref{eq:cal-weight} in the weighted derivation tree.

\medskip
We conclude by sketching how this algorithm generalises to one for general $\WLP$ programs. First, recall the theorem-proving and problem-solving perspectives on logic programming, as discussed in Appendix \ref{app:general_computation}. For $\WLP$, the theorem-proving aspect is prevalent in most applications (for example \cite{shay2010,eisner2007,durbin1998}). Within this perspective, the task is to compute the weight of the goal itself, rather than the weight of some substitution instance of the goal. To compute such term-matching weight semantics, say for a goal $A$ , one should first apply desaturation (as in Appendix \ref{app:general_computation}, see also \cite{BonchiZ15}) to the weighted saturated derivation tree $\intp{A}_{\saturate{\coalgU}}$, obtaining a new tree $T_A$. Concretely, $T_A$ is obtained from $\intp{A}_{\saturate{\coalgU}}$ by pruning any subtree which is rooted by substitutions other than the identity substitution.  

Second, one may tweak Algorithm \ref{alg:prune-ground-weighted-tree} and \ref{alg:calulate-pruned-weighted-tree} for ground $\WLP$ in order to deal with the substitution-nodes in `desaturated' weighted derivation trees. The term-matching weight semantics of $A$ can then be computed by applying the modified algorithms to $T_A$. Note that to guarantee termination of the computation, we require that for every atom $A$, there are only finitely many grounding $H\tau \ot B_1 \tau, \dots, B_k\tau$ of the clauses whose heads $H\tau$ is $A$. 

\section{Missing proofs}\label{app:missing-proofs}
\begin{prop}
The operation $\posw: \F\To \G$ defined in Definition \ref{def:possibleworldnat} is a natural transformation.
\end{prop}
\begin{proof}
The naturality for $\posw$ boils down to showing that for arbitrary $h\colon X\to Y$, $\G(h) \after \posw_X = \posw_Y \after \F(h) \colon \F(X) \to \G(Y)$. In this proof we fix an arbitrary function $h\colon X\to Y$. We start from an arbitrary $\varphi \in \F(X)$. The case where $\supp(\varphi) = \emptyset$ is easy. So we assume that $\varphi$ has a finite non-empty support $\supp(\varphi)$. On one hand, $\posw_X(\varphi)$ has support $\powset(\supp(\varphi))$, and for any $A\in \powset (\supp(\varphi))$, $\posw_X(\phi)(A) = \left( \prod_{a\in A} \varphi(a) \right) \cdot \left( \prod_{a\in \supp(\varphi)\setminus A} (1 - \varphi(a)) \right)$. Then the support of $\G(h)(\posw_X(\varphi))$ is $h[\supp(\posw_X(\varphi))] = h[ \powset (\supp(\phi)) ]$, and for each $B\in h[ \powset (\supp(\phi)) ]$, 
\begin{align*}
\G(h) (\posw_X(\varphi)) (B) & = \sum_{\substack{A\in \supp(\posw_X(\phi)) \\ h[A] = B}} \posw_X(\phi)(A) \\
& = \sum_{\substack{A\in \powset(\supp(\varphi)) \\ h[A] = B}} \posw_X(\phi)(A) \\
& = \sum_{\substack{ A\in \powset(\supp (\varphi)) \\ h[A] = B}} \left( \prod_{a\in A} \varphi(a) \cdot \prod_{a\in \supp(\varphi)\setminus A} (1 - \varphi(a)) \right) \numberthis \label{eq:posw-nat-fst-comp}
\end{align*}

On the other hand, $\F(h)(\varphi)$ has support $h[\supp(\varphi)]$, and for arbitrary $y\in h[\supp(\varphi)]$, $\F(h)(\varphi)(y) = \biglor_{x\in h^{-1}(y)} \varphi(x) = 1 - \prod_{x\in h^{-1}(y)}(1-\varphi(x))$. Then $\posw_Y(\F(h)(\varphi))$ has support $\powset(\supp(\F(h)(\varphi))) = \powset(h[\supp(\varphi)])$. For each $B' \in \powset(h[\supp(\varphi)])$,

\begin{equation}\label{eq:posw-nat-snd-comp}
\posw_Y(\F (h) (\varphi))(B')  = \prod_{y\in B'} \F(h)(\varphi)(y) \cdot \prod_{y\in h[\supp(\varphi)]\setminus B'} \left( 1 - \F(h)(\varphi)(y) \right)
\end{equation}

Note that the two supports $h[\powset(\supp(\phi))]$ and $\powset(h[\supp(\phi)])$ are equivalent: $B\in h[\powset(\dom(\varphi))]$ if and only if $\exists A\subseteq \dom(\varphi)$ such that $h[A] = B$; $B'\in \powset(h[\supp(\varphi)])$ if and only if $B'\subseteq h[\supp(\varphi)]$ if and only if $\exists A' \subseteq \supp(\varphi)$ such that $h[A'] = B'$.

We spell out the operation of $\F$ on $h$ in \eqref{eq:posw-nat-snd-comp}, and the steps \eqref{eq:equiv-pos-worlds-true}, \eqref{eq:dist-sum-prod} below will be explained in detail after the calculation:
\begin{align*}
& ~~\posw_Y(\F (h) (\varphi))(B') \\
= & \prod_{y\in B'} \left( 1 - \prod_{x\in h^{-1}(y)} (1 - \varphi(x)) \right) \cdot \prod_{y\in h[\supp(\varphi)]\setminus B'} \left( 1 - \left( 1 - \prod_{x\in h^{-1}(y)}  (1 - \varphi(x)) \right) \right) \\
= &  \prod_{y\in B'} \left( 1 - \prod_{x\in h^{-1}(y)} (1 - \varphi(x)) \right) \cdot \prod_{y\in h[\supp(\varphi)]\setminus B'} \prod_{x\in h^{-1}(y)}  (1 - \varphi(x)) \\
\eql{\eqref{eq:equiv-two-sides}} & \prod_{y \in B'} \left( \sum_{Z\subseteq_i h^{-1}(y)} \left( \prod_{x\in Z} \varphi(x) \cdot \prod_{x\in h^{-1}(y)\setminus Z} (1 - \varphi(x)) \right) \right) \cdot \prod_{\substack{y\in h[\supp(\varphi)]\setminus B' \\ x\in h^{-1}(y)}} (1 - \varphi(x)) \numberthis \label{eq:equiv-pos-worlds-true} \\
= & \sum_{\substack{A'\subseteq \supp(\varphi) \\ h[A'] = B'}} \left( \prod_{x\in A'} \varphi(x) \cdot \prod_{x\in h^{-1}[B] \setminus A'} (1- \varphi(x)) \right) \cdot \prod_{\substack{y\in h[\supp(\varphi)]\setminus B' \\ x\in h^{-1}(y)}} (1 - \varphi(x)) \numberthis \label{eq:dist-sum-prod} \\
= & \sum_{\substack{A'\subseteq \supp(\varphi) \\ h[A'] = B'}} \left( \prod_{x\in A'} \varphi(x) \cdot \prod_{x\in h^{-1}[B'] \setminus A'} (1- \varphi(x)) \right) \cdot \prod_{\substack{x\in \supp(\varphi) \\ x\not\in h^{-1}[B']}} (1 - \varphi(x)) \\
= & \sum_{\substack{A'\subseteq \supp(\varphi) \\ h[A'] = B'}} \left( \prod_{x\in A'} \varphi(x) \cdot \prod_{x\in h^{-1}[B'] \setminus A'} (1- \varphi(x)) \cdot \prod_{\substack{x\in \supp(\varphi) \\ x\not\in h^{-1}[B']}} (1 - \varphi(x)) \right) \\
= & \sum_{\substack{A'\subseteq \supp(\varphi) \\ h[A'] = B'}} \left( \prod_{x\in A'} \varphi(x) \cdot \prod_{x\in \supp(\varphi) \setminus A'} (1 - \varphi(x)) \right) \\
\eql{\eqref{eq:posw-nat-fst-comp}} & ~~\G(h)(\F(\phi))(B')
\end{align*}
In \eqref{eq:equiv-pos-worlds-true}, $\subseteq_i$ is the relation $\subseteq$ restricted to non-empty subsets. Formally, \eqref{eq:equiv-pos-worlds-true} is obtained via the following equation, for arbitrary set $X$, $\psi\in \F(X)$ and set $A\subseteq \supp(\psi)$:
\begin{equation}\label{eq:equiv-two-sides}
1 - \prod_{a\in A} (1 - \psi(a)) = \sum_{Z\subseteq_i A} \left(\prod_{a\in Z} \psi(a) \cdot \prod_{a\in A\setminus Z} (1 - \psi(a)) \right)
\end{equation}
Intuitively, one can think of $A$ as a set of independent events, and $\psi(a)$ as the probability of the event $a\in A$. Then the left-hand-side of \eqref{eq:equiv-two-sides} calculates $1$ minus the probability that none of $a\in A$ happens; the right-hand-side sums up the probabilities of all possible worlds in which some of $a\in A$ happens. Both calculate the same value, namely the probability that at least one $a\in A$ happens.

The expression \eqref{eq:dist-sum-prod} is obtained by expanding $\prod_{y\in B'} \sum_{Z\subseteq_i h^{-1}(y)} G(y, Z)$ into a summation, where $G(y, Z) = \prod_{x\in Z} \varphi(x) \cdot \prod_{x\in h^{-1}(y) \setminus Z} (1 - \phi(x))$. In the resulting summation, each summand chooses exactly one element from each $\sum_{Z\subseteq_i h^{-1}(y)} G(y, Z)$, for all $y\in B'$. Moreover, every such choice corresponds to exactly one $A'$ satisfying $h[A'] = B'$. So we have
\begin{align*}
\prod_{y\in B'} \sum_{Z\subseteq_i h^{-1}(y)} G(y, Z) & = \sum_{\substack{A'\subseteq \supp(\varphi) \\ h[A'] = B'}} \prod_{y\in B'} G(y, h^{-1}(y) \cap A') \\
& = \sum_{\substack{A'\subseteq \supp(\varphi) \\ h[A'] = B'}} \left( \prod_{y \in B'} \left( \prod_{x\in h^{-1}(y) \cap A'} \varphi(x) \cdot \prod_{x\in h^{-1}(y) \setminus (h^{-1}(y) \cap A')} (1 - \phi(x)) \right) \right) \\
& = \sum_{\substack{A'\subseteq \supp(\varphi) \\ h[A'] = B'}} \left( \prod_{x\in A'} \varphi(x) \cdot \prod_{x\in h^{-1}[B'] \setminus A'} (1 - \varphi(x)) \right) 
\end{align*}
and this is exactly the summation in \eqref{eq:dist-sum-prod}.

Therefore $(\G(h) \after \posw_X)(\phi)$ and $(\posw_Y \after \F(h))(\phi)$ have the same support, and both have the same values for arbitrary $B$ in the support.
\end{proof}

\end{document}